\documentclass{article}
\usepackage{amssymb,amsmath,latexsym,amscd,amsfonts,graphicx}

\newtheorem{thm}{Theorem}[section]
\newtheorem{proposition}{Proposition}[section]
\newtheorem{definition}{Definition}[section]

\newtheorem{lem}{Lemma}[section]

\newtheorem{remark}{Remark}[section]

\newenvironment{proof}{{\noindent{\bf Proof:}}}{$\hfill\Box$}

\def\tr{{\text{tr}}}
\def\dim{{\text{dim}}}

\def\v{\vert}
\def\r{\rangle}

\def\group{\mathbf{F}_{q}^+\!\rtimes\mathbf{F}_q^{\times}}
\def\near{\mathbf{H}^{+}\!\rtimes \mathbf{H}^{\times}}

\begin{document}

\title{The Quantum Double Model with Boundary: Condensations and Symmetries}

\author{Salman Beigi\\ {\it \small Institute for Quantum Information}\\[-1mm] {\it \small California Institute of
Technology, Pasadena, CA}\\[-1mm]{\it \small and}\\[-1mm]
{\it \small School of Mathematics}\\[-1mm]
{\it \small Institute for Research in Fundamental Sciences (IPM), Tehran, Iran}
\and Peter W. Shor \\ {\it \small Department of Mathematics}\\[-1mm] {\it \small Massachusetts Institute
of Technology, Cambridge, MA}
\and Daniel Whalen\\ {\it \small Department of Mathematics}\\[-1mm] {\it \small Massachusetts Institute
of Technology, Cambridge, MA}}

\maketitle


\begin{abstract} 
Associated to every finite group, Kitaev has defined the quantum double model for every orientable surface without boundary. In this paper,
we define boundaries for this model and characterize condensations; that is, we find all quasi-particle excitations (anyons) which disappear when they move to the boundary. We then consider two phases of the quantum double model corresponding to two groups with a domain wall between them, and study the tunneling of anyons from one phase to the other. Using this framework we discuss the necessary and sufficient conditions when two different groups give the same anyon types. As an application we show that in the quantum double model for $S_3$ (the permutation group over three letters) there is a chargeon and a fluxion which are not distinguishable. This group is indeed a special case of groups of the form of the semidirect product of the additive and multiplicative groups of a finite field, for all of which we prove a similar symmetry.
\end{abstract}

\section{Introduction}\label{sec:intro}

The quantum double model of Kitaev~\cite{kitaev} has been studied extensively in the recent years both from the point of view of quantum error correcting codes as well as non-abelian statistics. Bombin and Martin-Delgado~\cite{bombin} have defined a generalization of the Kitaev Hamiltonian and characterized condensations and confinements in their model. They have also studied their model with domain walls~\cite{nested}.  
Moreover, Levin-Wen model~\cite{levin-wen} can be considered as a generalization of the quantum double model for unitary tensor categories~\cite{aguado}, and provides us with a general approach to translate mathematical objects to physical concepts and vice versa. Kitaev and Kong (personal communication, 2009) have used this framework and defined the Levin-Wen model with boundaries and domain walls between two phases corresponding to two categories. In their model, the boundaries are parametrized by an algebra in the corresponding category. So by fixing an algebra, one can study edge excitations, condensations, and tunneling of quasi-particle excitations from one phase to another. However, the connection between this model and the work of Bombin and Martin-Delgado \cite{bombin, nested} is not clear. In particular, we do not know which algebras give the condensations of \cite{bombin}. In this paper, we define the quantum double model with boundary, and based on the recent work of 
Davydov~\cite{davydov} on the classification of algebras in group-theoretical modular categories, try to characterize possible condensations in this model. Before explaining our work and its consequences let us start with the example of the toric code.

Consider a planar square lattice and to each edge correspond a Hilbert space with basis elements indexed with $\mathbb{Z}_2$. The Hamiltonian is a summation of vertex and face operators which are defined in terms of $\sigma_x$ and $\sigma_z$ Pauli matrices (see~\cite{kitaev} for details). Then elementary excitations of the system correspond to certain chains of $\sigma_x$ and $\sigma_z$ operators (ribbon operators) which create two quasi-particles at the end points of the ribbon. A string of $\sigma_x$ operators gives \emph{magnetic charges} (fluxions), denoted by $m$, and a chain of $\sigma_z$ operators gives \emph{electric charges} (chargeons), denoted by $e$. Moreover, movement of quasi-particles is equivalent to extending the corresponding ribbons, and then their braidings can be defined. Simultaneous application of $\sigma_x$ and $\sigma_z$ chains corresponds to the fusion of $m$ and $e$: $\epsilon=m\otimes e$. 
These three particles together with the vacuum give the system of \emph{anyons} corresponding to group $\mathbb{Z}_2$; they are indeed the four irreducible representations of the quantum double of $\mathbb{Z}_2$, denoted by $D(\mathbb{Z}_2)$. 

Assume that there is a defect line in the lattice which divides the plane into two parts and so that the lattice on the right hand side is the dual of the lattice on the left. That is, on the left half-plane the vertex operators are defined in terms of $\sigma_x$ and on the right half-plane in terms of $\sigma_z$, and similarly for the face operators. (The vertex and face operators should be carefully defined on the defect line; details are given in the work of Kitaev and Kong and in~\cite{ising}.) Now consider an $m$ excitation on the left and by applying a chain of $\sigma_x$ operators move it to the right hand side. Due to the structure of the lattice, the string of $\sigma_x$ operators will change to $\sigma_z$ terms on the right, which means that $m$ becomes an $e$ on the right. Thus the operation of moving particles from the left half-plane to the right side exchanges $e$ and $m$ (while keeping the vacuum and $\epsilon$ unchanged). As a result, all braidings and fusions are symmetric with respect to the transposition $(e, m)$.

In this paper we generalize the above construction for every group. We consider a planar lattice with a defect line and define the Kitaev's Hamiltonian on the left and right half-planes corresponding to two groups $G$ and $G'$. The excitations on the two bulks again correspond to the representations of the quantum doubles of $G$ and $G'$, and ribbon operators create and move these quasi-particles. Therefore, if the movement of quasi-particles from the left half-plane to the right makes sense, i.e., a consistent definition of the Hamiltonian near the domain wall is available, we can study the tunneling of anyons from one phase to the other. 

Suppose $G'$ is the trivial single-element group. In this case there is no excitation on the right hand side, and indeed the domain wall turns into a boundary. The toric code with boundary ($G=\mathbb{Z}_2$) has been studied by Bravyi and Kitaev~\cite{bravyi}. They have considered two types of boundaries: the $z$-boundary and the $x$-boundary, and have shown that an $m$ ($e$) excitation disappears when it moves towards the $x$-boundary ($z$-boundary). In other words, anyons $m$ and $e$ get \emph{condensed} near the corresponding boundaries. Kitaev and Kong (personal communication, 2009) have generalized this idea and defined the Levin-Wen model with boundaries. In their model, the boundary is parametrized in terms of an algebra in the corresponding category. 

Here we define a boundary for the quantum double model in terms of a subgroup $K\subseteq G$ and a $2$-cocycle of $K$. Then we characterize the anyons that become condensed at the boundary. We finally by applying the \emph{folding idea} turn a domain wall between two phases $G$ and $G'$ into a boundary given by some $U\subseteq G\times G'$ and a $2$-cocycle, so the tunneling of anyons from one phase to another can be studied in terms of condensations.

Using this machinery, we study groups $G$ for which a symmetry similar to that of $\mathbb{Z}_2$ exists. That is, a transposition of a chargeon-fluxion pair, together with replacing each anyon with its charge conjugation, gives a symmetry of anyons corresponding to $G$. We show that all groups of the form $\group$, where $\mathbf{F}_q$ is the finite field with $q$ elements, have this property. Note that $q=2$ gives $\mathbb{Z}_2$ and $q=3$ corresponds to $S_3$, the permutation group over three letters.

The rest of this paper is organized as follows. In the following two sections, we review the basic ingredients of the quantum double of finite groups and the quantum double model. In Section~\ref{sec:boundary-1} we define a boundary for the Kitaev model that depends only on a subgroup (the corresponding $2$-cocycle is trivial) and compute the condensations. This construction is generalized in Section~\ref{sec:boundary-2}. In Section~\ref{sec:domain-wall} we consider a domain wall between two phases, and by applying the folding idea, turn it into a boundary. We then use all the previous results to find a non-trivial symmetry in the system of anyons of $\group$. In Appendix~\ref{app:b} we try to classify all groups for which a symmetry similar to that of $\group$ exists. Finial remarks and some open problems are discussed in Section~\ref{sec:conclusion}.

\section{Drinfeld double of a finite group}

Let us first fix some notations. $\mathbb{C}$ denotes the set of complex numbers and $\mathbb{C}^{\times}$ is its multiplicative group. $x^{\ast}$ is the complex conjugate of $x\in \mathbb{C}$ and $\v x\v^{2}=xx^{\ast}$. The identity element of a general group $G$ is denoted by $e$. For a subgroup $K$ of $G$ and $g\in K$, $Z_K(g)$ denotes the centralizer of $g$ in $K$: $Z_{K}(g)=\{h\in K:\, hg=gh\}$. We write $g\overset{K}\sim g'$ if there exists $h\in K$ such that $hgh^{-1}=g'$. When $K=G$ and there is no confusion we drop $K$ in these notations ($Z_{G}(g)=Z(g)$ and $g\sim g'$ means $g\overset{G}\sim g'$). The conjugacy class of $g\in G$ is denoted by $\overline{g}$ ($\overline{g}=\{hgh^{-1}:\, h\in G\}$). In this paper all representations are over complex numbers, and for a representation $\rho$ of a group, $\rho^{\ast}$ denotes its complex conjugate representation. $\tr_{\rho}(\cdot)$ is the character of $\rho$, and $\mathbf{1}$ denotes the trivial representation ($\tr_{\mathbf{1}}(\cdot) =1$). $\delta$ denotes the Kronecker delta function and for any relation $p$, $\delta_p =1$ if $p$ holds and otherwise $\delta_p = 0$. The size of a set $X$ is denoted by $\v X\v$. Finally, equivalence of categories, and isomorphism of groups and representations are shown by $\simeq$.

Although some of the results of this paper are stated in terms of category theory notions, a basic knowledge of the theory of anyons is enough to follow the proofs. For technical details we refer to~\cite{bakalov} and Appendix E of~\cite{anyons}.

\subsection{$D(G)$}

Let $G$ be a finite group. The quantum double or
Drinfeld double of $G$ denoted by $D(G)$, is a Hopf algebra containing $\mathbb{C}G$. $D(G)$ can be described by the $\mathbb{C}$-basis $\{gh^{\ast}: g,h \in G\}$ with the multiplication
\begin{align*}
(g_1h_1^{\ast})(g_2h_2^{\ast})=\delta_{h_2,\, g_2^{-1}h_1g_2}\, (g_1 g_2)h_2^{\ast},
\end{align*}
and the comultiplication
\begin{align}\label{eq:comulti}
\Delta(gh^{\ast}) = \sum_{h_1h_2=h} gh_1^{\ast}\otimes gh_2^{\ast}.
\end{align}
By $g\in D(G)$ we mean $g= \sum_h gh^{\ast}$, and $h^{\ast} = eh^{\ast}$ ($e$ is the identity of the group).
The unit of $D(G)$ is equal to $e=\sum_h e h^{\ast}$, the counit is given by $\varepsilon(g h^{\ast})=\delta_{h,e}$, and the antipode is
\begin{align*}
\gamma(g h^{\ast})= g^{-1} (gh^{-1}g^{-1})^{\ast}.
\end{align*}

\subsection{Representations of $D(G)$} \label{sec:irreps}

Consider an element $a\in G$ and let $\pi$ be a representation of $Z(a)$ over the vector space $W$ with the basis
$\{w_1, \dots, w_d\}$. Define the vector space $V_{( \overline{a}, \pi )}$ with the basis $\{\v b, w_i\r: b\in \overline{a},\, 1\leq i\leq d \}$. $V_{(\overline{a}, \pi)}$ is a representation of $D(G)$ as follows. For any $b\in \overline{a}$ fix $k_b\in G$ such that $b= k_b a
k_b^{-1}$. (Let $k_a = e$.) Observe that $k_{gbg^{-1}}^{-1}gk_b$ is always in $Z(a)$, and then for any $w\in W$, $b\in
\overline{a}$, and $gh^{\ast} \in D(G)$ define
\begin{align*}
g h^{\ast} \v b , w\r = \delta_{h, b}\, \v gbg^{-1} ,  \pi(k_{gbg^{-1}}^{-1}gk_b)\, w\r.
\end{align*}
It is easy to show that this action gives a representation of $D(G)$.  
$\chi_{(\overline{a}, \pi)}$, the character of this representation, is given by
\begin{align}\label{eq:dg-irreps}
\chi_{(\overline{a}, \pi)} (gh^{\ast}) =  \delta_{h\in \overline{a}}\, \delta_{gh, hg}\, \tr_{\pi} (k_h^{-1} g k_h).
\end{align}

If $\pi$ is an irreducible representation (irrep) of $Z(a)$, then the representation $V_{(\overline{a}, \pi)}$ of $D(G)$ is irreducible as well. Conversely, all irreps of $D(G)$ are of the above form and are indexed by conjugacy classes of $G$ and irreps of the centralizer of a fixed element in the corresponding conjugacy class (see for example \cite{bakalov}). 

The trivial representation of $D(G)$ is indexed by $\mathbf{0}=(e, \mathbf{1})$. Moreover, the \emph{(charge) conjugation} of $(\overline{a}, \pi)$, which we denote by $(\overline{a}, \pi)^{\vee}$, is isomorphic to $(\overline{a^{-1}},\pi^{\ast})$. The conjugacy class $\overline{a}$ of an irrep $(\overline{a}, \pi)$ is called its \emph{magnetic charge} and $\pi$ is its \emph{electric charge}. $(\overline{a}, \pi)$ is called a \emph{chargeon} if $a=e$ and a \emph{fluxion} if $\pi=\mathbf{1}$.

Irreducible representations of $D(G)$ are orthogonal to each other with respect to the following inner product:
\begin{align}\label{eq:inner-product-rep}
\langle \chi_1, \chi_2\rangle = \frac{1}{\vert G\vert} \sum_{g, h} \left(\chi_1(gh^{\ast})\right)^{\ast}\,  \chi_2(gh^{\ast}).
\end{align}
Then the multiplicity of the irrep $(\overline{a}, \pi)$ in the character $\chi$ is equal to $\langle \chi_{(\overline{a}, \pi)} , \chi\rangle$.

\subsection{Fusion rules}

Let $(\overline{a}, \pi)$ and $(\overline{a'}, \pi')$ be two irreps of $D(G)$. Then using the comultiplication~\eqref{eq:comulti}, $(\overline{a}, \pi)\otimes (\overline{a'}, \pi')$ is also a representation\footnote{The action of $gh^{\ast}$ on $v\otimes w$ is given by $\Delta(gh^{\ast}) v\otimes w$.} of $D(G)$ and is isomorphic to the direct sum of irreducible ones:
\begin{align*}
(\overline{a}, \pi)\otimes (\overline{a'}, \pi') \simeq \bigoplus_{(\overline{h}, \rho)}
N_{(\overline{a}, \pi)(\overline{a'}, \pi')}^{(\overline{h}, \rho)} (\overline{h},
\rho),
\end{align*}
where $ N_{(\overline{a}, \pi)(\overline{a'}, \pi')}^{(\overline{h}, \rho)}$ is a non-negative integer. To compute these numbers we may use the {\it Verlinde formula}.

Define the matrix $S$ whose rows and columns are indexed by irreps of $D(G)$ and
\begin{align}\label{eq:s-matrix}
S_{(\overline{a}, \pi)(\overline{a'}, \pi')} = \frac{1}{\v Z(a)\v \cdot \v Z(a')\v} \sum_{h:\, ha'h^{-1}\in Z(a)} \tr_{\pi}(ha'^{-1}h^{-1})\tr_{\pi'}(h^{-1}a^{-1}h).
\end{align}
Then $N_{XY}^Z$ can be computed in terms of $S$:
\begin{align}\label{eq:verlinde}
N_{XY}^{Z}=\sum_U \frac{S_{XU}S_{YU}S_{ZU}^{\ast}}{S_{\mathbf{0}U}},
\end{align}
where the summation runs over all irreps $U$, and $\mathbf{0}=(e, \mathbf{1})$ is the trivial representation.

\subsection{$\mathcal{Z}(G)$}\label{sec:z(g)}

$\mathcal{Z}(G)$ denotes the category of finite dimensional representations of $D(G)$ over complex numbers. Every object of $\mathcal{Z}(G)$ is isomorphic to a direct sum of simple objects, i.e., irreducible representations. $\mathcal{Z}(G)$ is a fusion category, where the fusion rules are given by the Verlinde formula. Moreover, the $R$-matrix 
\begin{align*} 
R=\sum_{g\in G} g^{\ast}\otimes g,
\end{align*}
defines the braiding $C_{X, Y}= PR: X\otimes Y \rightarrow Y\otimes X$ of two representations $X$ and  $Y$, where $P$ is the transposition of $X$ and $Y$. ($C_{X, Y}  v \otimes  w = \sum_{g} gw \otimes g^{\ast} v$.) $\mathcal{Z}(G)$ is a modular tensor category (see \cite{bakalov} for details).

\subsection{Example: $\mathcal{Z}(S_3)$}\label{sec:s3}

Let $G=S_3$ be the permutation group over three letters: $S_3= \langle \sigma, \tau: \sigma^2=\tau^3=e, \sigma \tau = \tau^{-1}\sigma \rangle $. $D(S_3)$ has eight irreducible representations described in the following table.
\begin{align*}
\begin{array}{|c|ccc|cc|ccc|}
\hline
      & A & B & C & D & E  & F & G & H   \\
\hline
\text{conjugacy class} & e & e & e & \overline{\sigma} & \overline{\sigma} &  \overline{\tau} & \overline{\tau} & \overline{\tau}   \\
 \text{irrep of the centralizer}     & \mathbf{1} & sign & \pi & \mathbf{1} & [-1] &  \mathbf{1} & [\omega] & [\omega^{\ast}]  \\
\hline
\end{array}
\end{align*}
Here $sign$ denotes the sign representation, $\pi$ is the two-dimensional representation of $S_3$, and $[-1]$, and $[\omega], [\omega^{\ast}]$ denote the non-trivial representations of $Z(\sigma)=\{e, \sigma\}$ and $Z(\tau)=\{e, \tau, \tau^{-1}\}$. The corresponding $S$-matrix is \begin{align}\label{eq:s-s3}
S = \frac{1}{6} \left(\begin{array}{ccc|cc|cccc}
      1 & 1 & 2 &  3 &  3 &  2 & 2 & 2   \\
      1 & 1 & 2 & -3 & -3 &  2 & 2 & 2  \\
      2 & 2 & 4  & 0 & 0 &  -2 & -2 & -2  \\
      \hline
      3 & -3 & 0 &  3 & -3 &  0 & 0 & 0 \\
      3 & -3 & 0  & -3 & 3  & 0 & 0 & 0 \\
      \hline
      2 & 2 & -2  & 0 & 0  & 4 & -2 & -2 \\
      2 & 2 & -2  & 0 & 0  & -2 & -2 & 4 \\
      2 & 2 & -2  & 0 & 0  & -2 & 4 & -2
\end{array}\right),
\end{align}
and then using the Verlinde formula~\eqref{eq:verlinde} the fusion rules can be computed.

{\small
\begin{align*}
\begin{array}{|c|ccc|cc|ccc|}
\hline
  \otimes    & A & B & C  & D & E  & F & G & H   \\
\hline
 A   & {\scriptstyle A} & {\scriptstyle B} & {\scriptstyle C} & {\scriptstyle D} & {\scriptstyle E} &  {\scriptstyle F} & {\scriptstyle G} & {\scriptstyle H}   \\
 B    & {\scriptstyle B} & {\scriptstyle A} & {\scriptstyle C} & {\scriptstyle E} & {\scriptstyle D} &  {\scriptstyle F} & {\scriptstyle G} & {\scriptstyle H}  \\
 C    & {\scriptstyle C} & {\scriptstyle C} & {\scriptstyle A\oplus B\oplus C}  & {\scriptstyle D\oplus E} & {\scriptstyle D\oplus E} &  {\scriptstyle G\oplus H} & {\scriptstyle F\oplus H} & {\scriptstyle F\oplus G}  \\
\hline
 D    & {\scriptstyle D} & {\scriptstyle E} & {\scriptstyle D\oplus E} &  {\scriptstyle A\oplus C\oplus F\oplus G\oplus H} & {\scriptstyle B\oplus C\oplus F\oplus G\oplus H} &  {\scriptstyle D\oplus E} & {\scriptstyle D\oplus E} &
 {\scriptstyle D\oplus E} \\
 E    & {\scriptstyle E} & {\scriptstyle D} & {\scriptstyle D\oplus E}  & {\scriptstyle B\oplus C\oplus F\oplus G\oplus H} & {\scriptstyle A\oplus C\oplus F\oplus G\oplus H}  & {\scriptstyle D\oplus E} & {\scriptstyle D\oplus E} & {\scriptstyle D\oplus E} \\
\hline
 F    & {\scriptstyle F} & {\scriptstyle F} & {\scriptstyle G\oplus H}  & {\scriptstyle D\oplus E} & {\scriptstyle D\oplus E}  & {\scriptstyle A\oplus B\oplus F} & {\scriptstyle H\oplus C} & {\scriptstyle G\oplus C} \\
 G    & {\scriptstyle G} & {\scriptstyle G} & {\scriptstyle F\oplus H}  & {\scriptstyle D\oplus E} & {\scriptstyle D\oplus E}  & {\scriptstyle H\oplus C} & {\scriptstyle A\oplus B\oplus G} & {\scriptstyle F\oplus C} \\
 H    & {\scriptstyle H} & {\scriptstyle H} & {\scriptstyle F\oplus G}  & {\scriptstyle D\oplus E} & {\scriptstyle D\oplus E}  & {\scriptstyle G\oplus C} & {\scriptstyle F\oplus C} & {\scriptstyle A\oplus B\oplus H} \\
\hline
\end{array}
\end{align*}
}

The are several symmetries in this fusion table. In particular, by exchanging $C$ and $F$ we obtain the same table. This fact can also be seen from the $S$-matrix~\eqref{eq:s-s3}. If we let $P$ to be the permutation matrix corresponding to the transposition $(C, F)$, then $PSP^{-1}=S$. On the other hand, by the Verlinde formula the fusion rules are computed in terms of $S$, so since $S$ is invariant under $P$, the fusion rules are also symmetric with respect to the transposition $(C, F)$. In Section~\ref{sec:3} we prove that this symmetry can be extended to an auto-equivalence of the whole category $\mathcal{Z}(G)$, i.e., the braidings are also symmetric.

\begin{figure}
\centerline{
\includegraphics[width=3.2in]{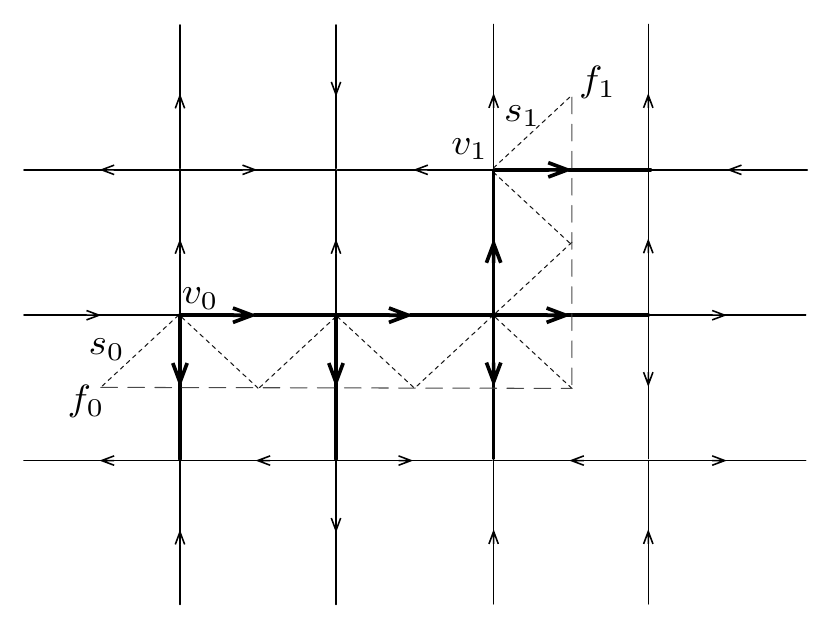}
\kern 0.6 in
\includegraphics[width=1.4in]{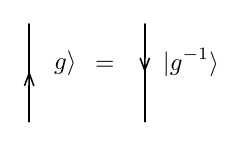}
}
\caption{A planer square lattice with directed edges. The direction of an edge can be reversed by changing the corresponding state according to the right figure. The pair of adjacent vertex $v_0$ and face $f_0$ consist a site. This site $s_0=(v_0, f_0)$ is depicted as a dotted line. A ribbon connecting two sites $s_0$ and $s_1$ is also shown. The corresponding ribbon operator acts on bold edges and is defined in Figure~\ref{fig:ribbon}.}
\label{fig:lattice}
\end{figure}

\section{The quantum double model}\label{sec:kitaev}

In this section we briefly discuss the main ingredients of the quantum double model \cite{kitaev}. For a detailed description we refer the reader to the original paper and \cite{bombin}.

\subsection{Kitaev's Hamiltonian}\label{sec:hamiltonian}

Consider a planar lattice with directed edges. We associate to each edge the Hilbert space $\mathbb{C}G$ with the orthonormal basis $\{\vert g\rangle:\, g\in G\}$ where $G$ is a finite group. For simplicity of presentation we assume that the direction of an edge can be reversed, and in this case we change the vector corresponding to that edge by sending $\vert g\rangle $ to $\vert g^{-1}\rangle$ (and extending linearly). A pair $s=(v, f)$ of adjacent vertex $v$ and face $f$ is called a \emph{site} and is depicted by a dotted line as in Figure~\ref{fig:lattice}. For any site $s=(v, f)$ we define the operators $A_s^g$  and $B_s^h$, $g,h\in G$, according to Figure~\ref{fig:a-b}. Observe that $A_s^g$ depends only on the vertex $v$ and not $f$, and furthermore, if $h$ is in the center of $G$, $B_s^h$ is independent of $v$. So if there is no ambiguity $A_s^g$ and $B_s^h$ are denoted by $A_v^g$ and $B_f^h$, and are called the vertex and face operators, respectively. The following relations are easy to verify:
\begin{align*}
A_s^g A_s^{g'} = A_s^{gg'},\\
(A_s^g)^{\dagger} = A_s^{(g^{-1})},\\
B_s^hB_s^{h'} = \delta_{h,h'} B_s^h,\\
(B_s^h)^{\dagger} = B_s^h,\\
A_s^gB_s^h = B_s^{(ghg^{-1})}A_s^g.
\end{align*}
These equations show that $gh^{\ast} \mapsto A_s^gB_s^h$ gives an isomorphism between the quantum double $D(G)$ and the algebra of operators acting on site $s$.

\begin{figure}
\centering{
\includegraphics[width=3 in]{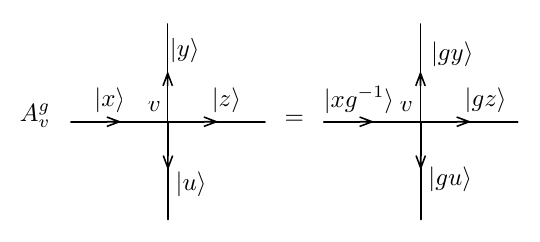}\\
\includegraphics[width=3.3 in]{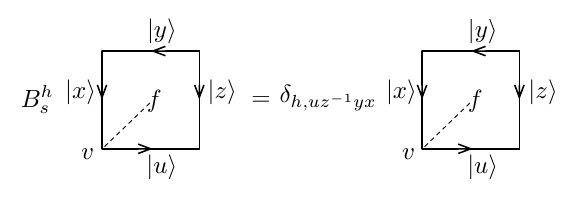}
}
\caption{Definition of operators $A_s^{g}$ and $B_s^{h}$. Note that if $h$ is in the center of $G$, $B_s^h$ depends only on the face $f$ and not vertex $v$. }
\label{fig:a-b}
\end{figure}

Observe that, for different sites $s\neq s'$ we have $[A_s^g, A_{s'}^{g'}]= [A_s^{g}, B_{s'}^{h'}]=[B_s^h, B_{s'}^{h'}]=0$. Define
$$A_v=A_s=\frac{1}{\vert G\vert} \sum_{g\in G} A_s^g$$
and
$$B_f=B_s= B_s^e.$$
Operators $A_v$ and $B_f$ are projections and pairwise commute. Then consider the Hamiltonian
\begin{align*}
H_G= - \sum_v A_v -\sum_f B_f,
\end{align*}
where the summations run over all vertices $v$ and faces $f$. Since all terms of the Hamiltonian commute, the ground state of $H_G$ is a state $\vert \psi\rangle$ such that $A_v\vert \psi\rangle = B_f\vert \psi\rangle =\vert \psi\rangle$. For a planar lattice the ground state is unique and can be explicitly computed~\cite{kitaev}. Nevertheless, here we are interested in elementary excitations.

\subsection{Ribbon operators}

A ribbon $\xi$ in the lattice is a sequence of ``adjacent" sites connecting two sites $s_0$ and $s_1$ as in Figure~\ref{fig:lattice}. From now on we always assume that $s_0$ is the \emph{starting} site of $\xi$ and $s_1$ is the \emph{ending} site. (We also assume that ribbons avoid self-crossing.) For any ribbon $\xi$ and $g,h\in G$ the \emph{ribbon operator} $F_{\xi}^{h,g}$ is defined as in Figure~\ref{fig:ribbon}. It is easy to see that
\begin{align}
F_{\xi}^{h,g} F_{\xi}^{h',g'} &=\delta_{g,g'} F_{\xi}^{hh',g}\label{eq:f-multi},\\
(F_{\xi}^{h,g})^{\dagger} &=F_{\xi}^{h^{-1},g}.
\end{align}
Moreover, for every site $t$ different from $s_0$ and $s_1$, $[F_{\xi}^{h,g}, A_t^{k}]=[F_{\xi}^{h,g}, B_t]=0$, and we have
\begin{align}
A_{s_0}^k F_{\xi}^{h,g} &= F_{\xi}^{khk^{-1}, kg} A_{s_0}^k \label{eq:com-0-a},\\
B_{s_0}^{k} F_{\xi}^{h,g} &= F_{\xi}^{h,g} B_{s_0}^{kh}, \label{eq:com-0-b}
\end{align}
and
\begin{align}
A_{s_1}^k F_{\xi}^{h,g} &= F_{\xi}^{h, gk^{-1}} A_{s_1}^k \label{eq:com-1-a},\\
B_{s_1}^{k} F_{\xi}^{h,g} &= F_{\xi}^{h,g} B_{s_1}^{g^{-1}h^{-1}gk}. \label{eq:com-1-b}
\end{align}

\begin{figure}
\centering{
\includegraphics[width=5 in]{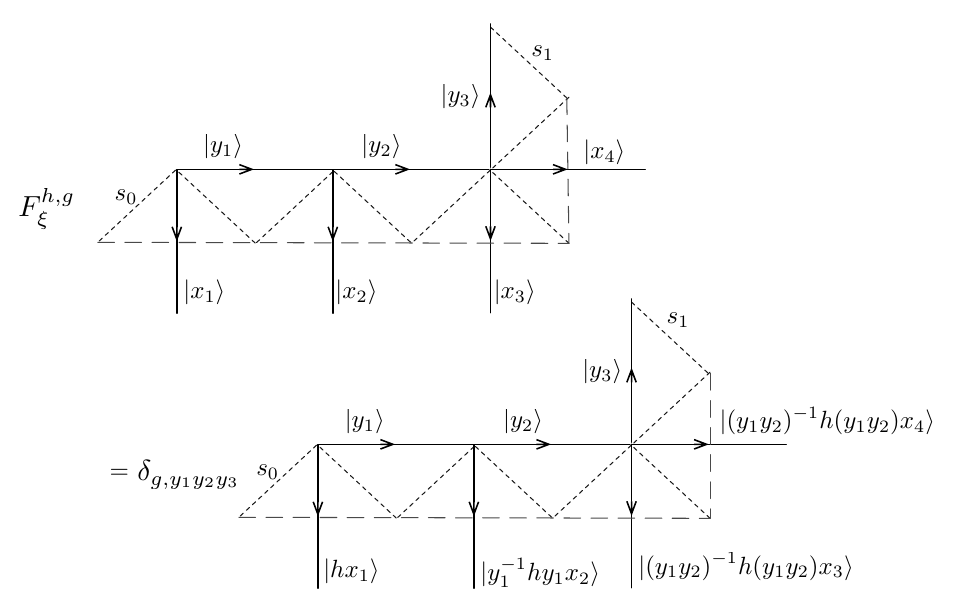}
}
\caption{Definition of the ribbon operator $F_{\xi}^{h,g}$. The ribbon $\xi$ connects the starting site $s_0$ to the ending site $s_1$.}
\label{fig:ribbon}
\end{figure}

\subsection{Elementary excitations}\label{sec:elem-excit}

Let $\vert \psi\rangle$ be the ground state of the Hamiltonian $H_G$. Fix a ribbon $\xi$ which connects two sites $s_0$ and $s_1$. Since $F_{\xi}^{h^{-1},g}$ commutes with all terms of $H_G$ except the terms at sites $s_0$, $s_1$, the state $\vert \psi^{h,g} \rangle = F_{\xi}^{h^{-1},g}\vert \psi\rangle$ satisfies all constraints of the Hamiltonian except the ones at $s_0$ and $s_1$. Moreover, $\vert \psi^{h,g}\rangle$ does not depend on $\xi$, but only on the end points $s_0, s_1$; that is, if $\xi'$ is another ribbon with the same end points, $\vert \psi^{h,g}\rangle=F_{\xi}^{h^{-1},g}\vert \psi\rangle = F_{\xi'}^{h^{-1},g}\vert \psi\rangle$,~\cite{kitaev}. Therefore, applying $F_{\xi}^{h,g}$ on the ground state can be thought of as creating a pair of quasi-particles at sites $s_0$ and $s_1$. Thus movement of such quasi-particles is equivalent to extending the corresponding ribbons, and then their braidings can be defined. Furthermore, to fuse two quasi-particles of this form we can simply move them to the same site. As a result, the set of these quasi-particles describes a system of \emph{anyons}.

Now the question is whether there exists other elementary excited states or not. The space of quasi-particle excitations living at $s_0$ and $s_1$ is equal to 
$$\mathcal{L}(s_0, s_1) = \{\vert v\rangle : \, A_t\vert v\rangle =B_t\vert v\rangle =\vert v\rangle \text{ for all } t\neq s_0, s_1 \}.$$
By the above argument $\vert \psi^{h,g}\rangle$ belongs to $\mathcal{L}(s_0, s_1)$, and it is proved in \cite{kitaev} that  $\mathcal{L}(s_0, s_1)$ is spanned by states $\vert \psi^{h,g}\rangle$. Therefore, all excitations in this system can be obtained by applying ribbon operators on the ground state.
The inner product in this space is computed in the following lemma.

\begin{lem}\label{lem:inner-product} 
\begin{align}\label{eq:inner-product}
\langle \psi^{h,g}\vert \psi^{h',g'}\rangle = \frac{1}{\vert G\vert}\, \delta_{h,h'} \delta_{g,g'}.
\end{align}
\end{lem}

\begin{proof}
$\langle \psi^{h,g}\vert \psi^{h',g'}\rangle = \langle \psi \vert (F_{\xi}^{h^{-1}, g})^{\dagger} F_{\xi}^{h'^{-1}, g'} \vert \psi\rangle
= \delta_{g, g'} \langle \psi \vert F_{\xi}^{hh'^{-1}, g} \vert \psi \rangle.$
So it suffices to show that 
\begin{align}\label{eq:inner-product-f-4}
\langle \psi \vert F_{\xi}^{h, g} \vert \psi\rangle = \frac{1}{\vert G\vert}\delta_{h, e}.
\end{align}
Using \eqref{eq:com-1-b} we have
\begin{align*}
\langle \psi \vert F_{\xi}^{h, g} \vert \psi\rangle  = \langle \psi \vert B_{s_1}^e  F_{\xi}^{h, g} \vert \psi\rangle 
 = \langle \psi \vert F_{\xi}^{h, g} B_{s_1}^{g^{-1}h^{-1}g} \vert \psi\rangle.  
\end{align*}
Thus $\langle \psi \vert F_{\xi}^{h, g} \vert \psi\rangle  = 0$ if $h\neq e$. 
Now by \eqref{eq:com-1-a} and $A_{s_1}^k A_{s_1} = A_{s_1}$ we obtain
\begin{align*}
\langle \psi \vert F_{\xi}^{e, g} \vert \psi\rangle  & =  \langle \psi \vert A_{s_1}^k F_{\xi}^{e, g} \vert \psi\rangle \\
& = \langle \psi \vert F_{\xi}^{e, gk^{-1}} A_{s_1}^k \vert \psi\rangle\\
& = \langle \psi \vert F_{\xi}^{e, gk^{-1}} \vert \psi\rangle.   
\end{align*}
Then for every $g, g'$ we have $\langle \psi \vert F_{\xi}^{e, g} \vert \psi\rangle = \langle \psi \vert F_{\xi}^{e, g'} \vert \psi\rangle  $. On the other hand, $\sum_{g\in G} F_{\xi}^{e, g} $ is the identity operator, so we are done. 

\end{proof}

\subsection{Anyon-types}

The only remaining question is to find different types of anyons. Let $t$ be a site different from $s_0, s_1$, and let $\zeta$ be a \emph{closed ribbon} which encircles $s_0$ but not $s_1$, and both of whose end points are $t$. To characterize an unknown excitation sitting at $s_0$, we can create a particle-antiparticle pair at $t$, move one of them along $\zeta$ and rotate it around $s_0$, and finally by measuring the vertex and face operators at $t$ check whether the pair (after braiding) fuses to vacuum or not. We can identify the anyon at $s_0$ by repeating this process for different particle-antiparticle pairs that we create at $t$. 

Of course, if there is no excitation at $s_0$, we expect that after the rotation, the particle-antiparticle pair always fuses to vacuum. Mathematically, creating this pair and moving one of them, correspond to applying some ribbon operator along $\zeta$. Moreover, if the pair fuses to vacuum, this ribbon operator should not create any excitation at $t$. This means that, besides the vertex and face operators along $\zeta$, this ribbon operator must commute with $A_t$ and $B_t$ as well. 
Letting $\mathcal{F}_{\zeta}$ to be the algebra of ribbon operators $F_{\zeta}^{h,g}$, $g, h\in G$, we conclude that the anyon-types are characterized by the subalgebra $\mathcal{K}_{\zeta} \subseteq \mathcal{F}_{\zeta}$ of ribbon operators which commute with $A_t$ and $B_t$:
\begin{align*}
\mathcal{K}_{\zeta} = \{T\in \mathcal{F}_{\zeta}:\, [T, A_t] = [T, B_t]=0  \}.
\end{align*}
In \cite{bombin} for any irreducible representation $X$ of $D(G)$ a ribbon operator $T^{X} \in \mathcal{K}_{\zeta}$ is defined, and it is proved that $\mathcal{K}_{\zeta}$ is generated by these operators and the following equations hold:
\begin{align*}
(T^{X})^{\dagger} &= T^X,\\
T^X T^{Y} & = \delta_{X, Y} T^{X},\\
\sum_X T^X & = I.
\end{align*}
As a result, anyon-types are in one-to-one correspondence with irreps of $D(G)$. Indeed, the set of projections $T^X$ decompose the space of excitations 
\begin{align}\label{eq:x-decomp-l}
\mathcal{L}(s_0, s_1) = \bigoplus_X T^X \mathcal{L}(s_0, s_1),
\end{align}
and $\vert v\rangle \in T^X \mathcal{L}(s_0, s_1)$ is a state of an anyon of type $X$.

The decomposition \eqref{eq:x-decomp-l} can also be derived from another point of view. 
Using the commutation relations \eqref{eq:com-0-a} and \eqref{eq:com-0-b}, it is easy to see that
\begin{align*}
A_{s_0}^k \vert \psi^{h,g}\rangle & =\vert \psi^{khk^{-1}, kg}\rangle, \\
B_{s_0}^{l}\vert \psi^{h,g}\rangle &=\delta_{l, h} \vert \psi^{h,g}\rangle.
\end{align*}
As we mentioned in Section \ref{sec:hamiltonian} the algebra generated by operators $A_{s_0}^k$ and $B_{s_0}^l$ is isomorphic to $D(G)$. Then the above equations define a representation of $D(G)$ on the space $\mathcal{L}(s_0, s_1)$. This representation, however, is equivalent to the regular representation of $D(G)$; that is, by sending $\vert \psi^{h,g}\rangle$ to $h^{\ast}g$, the action of $D(G)$ is given by multiplication from left. On the other hand, decomposing the regular representation of $D(G)$ into irreducible ones, we obtain all irreps of $D(G)$ as a summand, i.e., as representations of $D(G)$ we have
$$\mathcal{L}(s_0, s_1) \simeq  \bigoplus_X \,(\dim X) V_X,$$
where the direct sum runs over irreps of $D(G)$ and $V_X$ is the vector space corresponding to irrep $X$. Again we refer to \cite{bombin} for an explicit description of this isomorphism in terms of basis vectors, which indeed is the same as the decomposition \eqref{eq:x-decomp-l}:  $T^X \mathcal{L}(s_0, s_1)\simeq (\dim X) V_X$. 

As a summary, anyon-types in this model are described by simple objects of $\mathcal{Z}(G)$. Moreover, given a subspace of elementary excitations $W\subseteq \mathcal{L}(s_0, s_1)$, to find the types of anyons in $W$ we can compute the character of the representation $W$ of $D(G)$, and then decompose it into irreducible ones. We will use this method in Sections \ref{sec:boundary-1} and \ref{sec:boundary-2} to find condensed anyons.

\begin{remark} Here the representation of $D(G)$ on the space $\mathcal{L}(s_0, s_1)$ is defined based on the action of the vertex and face operators at site $s_0$, and another representation can be found by considering those operators at $s_1$. However, it is not hard to see that these two representations of $D(G)$ are charge conjugations of each other and it does not matter which representation is picked to find anyon-types. Indeed, the anyon sitting at $s_1$ is the charge conjugation of the one at $s_0$.    
\end{remark}

\section{The quantum double model with boundary I}\label{sec:boundary-1}

In this section we define a boundary for the quantum double model and compute the corresponding condensation. This model will be generalized in the following section. 

Instead of a planar lattice, consider a lattice defined on a half-plane as in Figure~\ref{fig:lattice-boundary}.  Again $\mathbb{C}G$ is the Hilbert space associated with \emph{internal} edges, however, the Hilbert space of the boundary edges is $\mathbb{C}K$ where $K\subseteq G$ is a fixed subgroup. The vertex and face operators corresponding to internal sites are defined as before. But for a boundary site $s=(v,f)$ the vertex operator is
$$A_s^K=\frac{1}{\vert K\vert }\sum_{k\in K} A_s^k.$$
Also the fact that the corresponding Hilbert space to a boundary edge is $\mathbb{C}K \subseteq \mathbb{C}G$, can be captured by considering the projection onto this subspace 
$$B_s^K= \sum_{k\in K} B_s^k.$$
Now define the Hamiltonian
\begin{align} \label{eq:hamiltonian-g-k}
H_{G,K} = -\sum_{v} A_v - \sum_{f} B_f - \sum_{s} (A_s^K + B_s^K),
\end{align}
where the summations run over internal vertices and faces $v, f$ and boundary sites $s$. Observe that similar to $H_G$, all terms of $H_{G, K}$ commute.

By the same reasoning as before the bulk excitations can be created and moved by applying ribbon operators, and they are in one-to-one correspondence with irreps of $D(G)$. Some of these quasi-particles, however, may disappear when they move to the boundary; that is, because the local terms of the Hamiltonian on the boundary are different from internal terms, an excitation which violates some of the internal constraints, may become a ground state when it moves to a boundary site.

\begin{figure}
\centering{
\includegraphics[width=3.2 in]{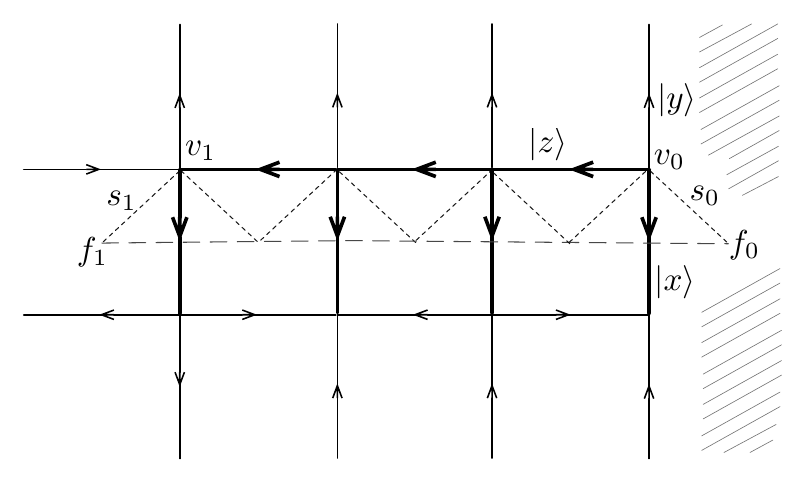}
}
\caption{A planar lattice with boundary. A ribbon connects the boundary site $s_0=(v_0, f_0)$ to the internal site $s_1=(v_1, f_1)$. Here, operators $A_{s_0}^{k}$ and $B_{s_0}^k$ are defined the same as before: $A_{s_0}^k \vert x, y, z\rangle = \vert kx, ky, kz\rangle$ and $B_{s_0}^k \vert x\rangle = \delta_{k, x}\vert x\rangle$.}
\label{fig:lattice-boundary}
\end{figure}

\subsection{Condensations}

Fix a ribbon $\xi$ which connects a boundary starting site $s_0$ to an internal site $s_1$ as in Figure~\ref{fig:lattice-boundary}. Let $\mathcal{C}_{\xi}\subseteq  \mathcal{F}_{\xi}$ to be the subalgebra of operators that commute with both $A_{s_0}^K$ and $B_{s_0}^K$:
\begin{align*}
\mathcal{C}_{\xi} = \{T\in \mathcal{F}_{\xi}:\, [T, A_{s_0}^K] = [T, B_{s_0}^K] =0\}.
\end{align*}
Then for every $T\in \mathcal{C}_{\xi}$, by applying $T$ on the ground state of $H_{G, K}$, we generate a quasi-particle at $s_1$, but no excitation at $s_0$. It means that the quasi-particle at $s_1$ disappears when it moves to the boundary site $s_0$. In this case we say that this excitation gets condensed at the boundary. So to classify condensations we should find the algebra $\mathcal{C}_{\xi}$.   

Let $T = \sum_{h,g} c_{h, g} F_{\xi}^{h,g}$ be in $\mathcal{C}_{\xi}$. Using the commutation relations\footnote{Although these equations are given for an internal site, it is not hard to see that they also hold for boundary sites.} \eqref{eq:com-0-a} and \eqref{eq:com-0-b}, $[T, A_{s_0}^K]=0$ is equivalent to 
$$c_{khk^{-1}, kg} = c_{h,g },$$
for every $k\in K$, and $[T, B_{s_0}^K]=0$ if and only if $c_{h, g} = 0$, for $h\notin K$. These two relations completely characterize $\mathcal{C}_{\xi}$ as follows.

For every $k\in K$ and $g\in G$ define
\begin{align}
T^{k,g} =\sum_{l\in K} F_{\xi}^{lkl^{-1}, lg^{-1}}.
\end{align}

\begin{proposition}\label{prop:rel-gen-1}
The algebra $\mathcal{C}_{\xi}$ is spanned by operators $T^{k,g}$, $k\in K$, $g\in G$, and the following equations hold.
\begin{enumerate}
\item $T^{k, gm} = T^{mkm^{-1}, g}$, for every $m\in K$.
\item $T^{k,g}T^{k',g'}=0$ if $gK\neq g'K$.
\item $T^{k,g}T^{k',g} = T^{kk', g}$.
\item $(T^{k, g})^{\dagger} = T^{k^{-1}, g}$.
\end{enumerate}
\end{proposition}

\begin{proof} The proof is straightforward and is left to the reader. (We will prove a generalization of this proposition in the next section.)

\end{proof}

Now in order to find the type of condensed anyons we use the idea of Section \ref{sec:elem-excit} and compute the representation of $D(G)$ induced by $\mathcal{C}_{\xi}$. Let $\vert \psi_K\rangle$ be the ground state of $H_{G,K}$. For $k\in K$ and $g\in G$ define $\vert \psi_K^{k,g}\rangle = T^{k,g} \vert \psi_K\rangle$ and let $\mathcal{A}(K)$ be the span of these vectors. Since the operators $T^{k, g}$ commute with $A_{s_0}^K$ and $B_{s_0}^K$, we have $A_{s_0}^K\vert \psi_K^{k,g}\rangle= B_{s_0}^K\vert \psi_K^{k,g}\rangle= \vert \psi_K^{k,g}\rangle$. But using \eqref{eq:com-1-a} and \eqref{eq:com-1-b}
\begin{align*}
A_{s_1}^h \vert \psi_K^{k,g}\rangle  & = \vert \psi_K^{k, hg}\rangle, \\
B_{s_1}^{h} \vert \psi_K^{k, g}\rangle &  = \delta_{h, gkg^{-1}} \vert \psi_K^{k, g}\rangle, 
\end{align*} 
and we obtain a representation of $D(G)$ on $\mathcal{A}(K)$. To compute the character of this representation we need to fix a basis. Assume that $\vert G\vert / \vert K\vert = r$ and $G = g_1K \cup \dots \cup g_rK$. Then by Proposition \ref{prop:rel-gen-1} the $\vert G\vert$ states $\vert \psi_K^{k, g_i}\rangle$, $k\in K$, $i=1, \dots , r$, span $\mathcal{A}(K)$. We have 
\begin{align}
\langle \psi_K^{k, g_i} \vert \psi_K^{k', g_j}\rangle  & = \langle \psi_K \vert (T^{k, g_i})^{\dagger} T^{k', g_j}\vert \psi_K\rangle  \nonumber\\
& = \delta_{i,j} \langle \psi_K \vert T^{k^{-1}k', g_i}\vert \psi_K\rangle \nonumber \\
& = \delta_{i,j} \sum_{m\in K} \langle \psi_K\vert F_{\xi}^{mk^{-1}k'm^{-1}, mg_i^{-1}}\vert \psi_K \rangle \nonumber \\
& = \frac{\vert K\vert}{\vert G\vert}\, \delta_{i,j} \delta_{k,k'}, \label{eq:in-pro-2}
\end{align}
where in the last line we use \eqref{eq:inner-product-f-4} which still holds even considering a boundary. So $\{\sqrt{r} \vert \psi_K^{k, g_i}\rangle:\,  k\in K, i=1, \dots, r   \}$ is an orthonormal basis for $\mathcal{A}(K)$.

Now we are ready to compute $\chi_{\mathcal{A}(K)}$ the character of the representation $\mathcal{A}(K)$. For $g, h\in G$ we have
\begin{align*}
\chi_{\mathcal{A}(K)} (hg^{\ast}) & =  r \sum_{k\in K} \sum_{i=1}^r \,\langle \psi_K^{k, g_i} \vert hg^{\ast} \vert \psi_K^{k, g_i}\rangle \\
& =  r \sum_{k\in K} \sum_{i=1}^r \, \delta_{g, g_ikg_i^{-1}} \langle \psi_K^{k, g_i} \vert h\vert \psi_K^{k, g_i}\rangle \\
& = r \sum_{k\in K} \sum_{i=1}^r \,  \delta_{g, g_ikg_i^{-1}}  \langle \psi_K^{k, g_i}  \vert \psi_K^{k, hg_i}\rangle.
\end{align*}
Let $hg_i = g_{\epsilon(i)} k_i$ where $k_i\in K$ and $1\leq \epsilon(i)\leq r$. Thus by Proposition \ref{prop:rel-gen-1} we have
\begin{align*}
\chi_{\mathcal{A}(K)} & = r \sum_{k\in K} \sum_{i=1}^r \,  \delta_{g, g_ikg_i^{-1}}  \langle \psi_K^{k, g_i}  \vert \psi_K^{k_ikk_i^{-1}, g_{\epsilon(i)}}\rangle \\
&= r \sum_{k\in K} \sum_{i=1}^r \,  \delta_{g, g_ikg_i^{-1}} \,\delta_{k, k_ikk_i^{-1}} \delta_{i, \epsilon(i)}\frac{\vert K\vert}{ \vert G\vert} \\
& =  \sum_{k\in K} \sum_{i=1}^r \,  \delta_{g_i^{-1}gg_i, k}\, \delta_{g_i^{-1} h g_i , k_i} \delta_{kk_i ,k_ik}   \\
& =  \sum_{k\in K} \sum_{i=1}^r \,  \delta_{g_i^{-1}gg_i, k}\, \delta_{g_i^{-1} h g_i , k_i} \delta_{gh ,hg}   \\
& =   \delta_{gh ,hg} \sum_{i=1}^r \delta_{g_i^{-1} gg_i \in K} \,\delta_{g_i^{-1} h g_i \in K},
\end{align*}
and then 
\begin{align}\label{eq:char-phi=1}
\chi_{\mathcal{A}(K)} (hg^{\ast})  = \frac{1}{\vert K\vert}  \delta_{gh ,hg}\sum_{x\in G} \delta_{xgx^{-1}\in K } \, \delta_{xhx^{-1} \in K}.
\end{align}
To compute the set of condensed anyons we just need to decompose this character into irreducible ones.

Let us give some examples. Let $K=G$. We have 
\begin{align}
\chi_{A(G)}(hg^{\ast}) =  \frac{1}{\vert G\vert } \delta_{gh,hg} \sum_{x\in G}  1 = \delta_{gh, hg}.
\end{align}
Then using \eqref{eq:dg-irreps} it is easy to see that $\chi_{A(G)} = \sum_{\overline{a}} \chi_{(\overline{a}, \mathbf{1})}$. As a result, in this case the condensation corresponds to all fluxions.
Another example is $K=\{e\}$. We have 
$$\chi_{A(\{e\})}(hg^{\ast})= \vert G\vert \delta_{g, e}\delta_{h,e} = \delta_{g,e} \tr_{\rho} (h),$$ 
where $\rho$ denotes the regular representation of $G$. Therefore, $\chi_{A(\{e\})}=\sum_{\pi} \dim \pi\, \chi_{(e, \pi)}$ where the summation is over irreps of $G$, and condensed anyons are all chargeons.

In general, an anyon indexed by $(\overline{a}, \pi)$ is condensed if $\chi_{\mathcal{A}(K)}$ and $\chi_{(\overline{a}, \pi)}$ are not orthogonal: $\langle \chi_{\mathcal{A}(K)}, \chi_{(\overline{a}, \pi)} \rangle >0$ (see \eqref{eq:inner-product-rep} for the definition of the inner product of characters). Expanding $\langle \chi_{\mathcal{A}(K)}, \chi_{(\overline{a}, \pi)} \rangle $, it is easy to see that the condensed anyons in our model coincide with those in the model proposed in \cite{bombin}. Bombin and Martin-Delgado have defined a variation of the quantum double model in which two subgroups $N\subseteq M$ are involved, and described the necessary and sufficient condition for an anyon $(\overline{a}, \pi)$ to be condensed. This condition in the case where $M=N=K$ is the same as our constraint $\langle \chi_{\mathcal{A}(K)}, \chi_{(\overline{a}, \pi)} \rangle >0$ which means that these two models characterize similar condensations.\footnote{Bombin and Martin-Delgado \cite{bombin} have assumed that $N$ is a normal subgroup of $G$, however, it seems that the normality of $N$ in $M$ is enough.}

\begin{remark} In the example of $K=G$ the condensation consists of all fluxions. Since some of these condensed anyons may have non-trivial braidings, we should have confinements in this case. Confinements have been discussed in~\cite{bombin} and classified in some cases. 
\end{remark}

\section{The quantum double model with boundary II}\label{sec:boundary-2}

We now generalize the boundary defined in the previous section. Here besides a subgroup $K$, the boundary depends on a $2$-cocyle of $K$ as well.

\subsection{$2$-cocycles}

Let $\varphi:K\times K \rightarrow \mathbb{C}^{\times}$ be a function such that for every $k,l, m\in K$
\begin{align}\label{eq:cocycle}
\varphi(kl, m) \varphi(k, l) = \varphi(k, lm) \varphi(l, m).
\end{align}
Then  $\varphi$ is called a $2$-cocycle of $K$.  

Every $2$-cocycle comes from a \emph{projective representation} and vise versa.
A representation of $K$ on the vector space $V$ is indeed a homomorphism $K \rightarrow \text{GL}(V)$, where $\text{GL}(V)$ is the group of invertible linear transformations of $V$. A projective representation is a homomorphism 
$$\rho: K \rightarrow \text{PGL}(V),$$ 
where $\text{PGL}(V) = \text{GL}(V)/ \mathbb{C}$ is the quotient of $\text{GL}(V)$ modulo scalers. Thus every representation, by composing with the map 
$\Pi : \text{GL}(V) \rightarrow \text{PGL}(V),$ 
gives a projective representation, but the converse does not hold. However, a projective representation provides us with a $2$-cocycle: 
for every $k\in K$ fix $L(k) \in \text{GL}(V)$ such that $\Pi(L(k)) = \rho(k)$ ($L$ is a lifting of $\rho$). So $\Pi(L(kl)) = \Pi ( L(k)L(l) )$ and then there exists $\varphi(k, l) \in \mathbb{C}$ such that 
\begin{align}\label{eq:lifting}
\varphi(k, l) L(kl) =  L(k) L(l).
\end{align} 
It is easy to see that this function $\varphi$ satisfies \eqref{eq:cocycle} and is a $2$-cocycle.

Conversely, every $2$-cocycle $\varphi$ corresponds to a projective representation. Let $V=\mathbb{C}K$ and define $L(k) \vert l\rangle = \varphi (k, l) \vert kl \rangle$. Then  $\rho(k) = \Pi(L(k))$ gives a projective representation and~\eqref{eq:lifting} holds.

Let $\alpha: K \rightarrow \mathbb{C}^{\times}$ be an arbitrary function. Then $L'(k) = \alpha(k) L(k)$ also satisfies $\Pi(L'(k)) = \rho(k)$, and defines another $2$-cocycle $\varphi'$ corresponding to the same projective representation. We call such two $2$-cocycles $\varphi$ and $\varphi'$ \emph{equivalent}. More precisely, $\varphi$ and $\varphi'$ are equivalent if there exists $\alpha$ such that 
$$\varphi'(k,l) = \alpha{(kl)}^{-1} \alpha(k)\alpha(l)\varphi(k, l).$$ 
The set of $2$-cocycles of $K$ up to the above equivalency is denoted by $H^{2}(K, \mathbb{C}^{\times})$.

\begin{lem} \label{lem:2-cocycle}
Every $2$-cocycle is equivalent to one with the following properties.
\begin{enumerate}
\item $\varphi(e, k) = \varphi(k, e) =1$
\item $\varphi(k, k^{-1}) =1$
\item $\vert \varphi(k, l) \vert  =1$
\item $\varphi (k^{-1}, l^{-1}) = \varphi(l, k)^{-1}$.
\end{enumerate}

\end{lem}

\begin{proof}
Every $2$-cocycle corresponds to a projective representation. Consider a lifting of this projective representation such that $L(e)= I$, $L(k^{-1}) = L(k)^{-1}$, and $det\, L(k)=1$ for every $k\in K$. Then the above equations follow from \eqref{eq:lifting}.

\end{proof}

For simplicity from now on we assume that $\varphi$ satisfies the properties of this lemma.

\begin{figure}
\centering{
\includegraphics[width=3.2 in]{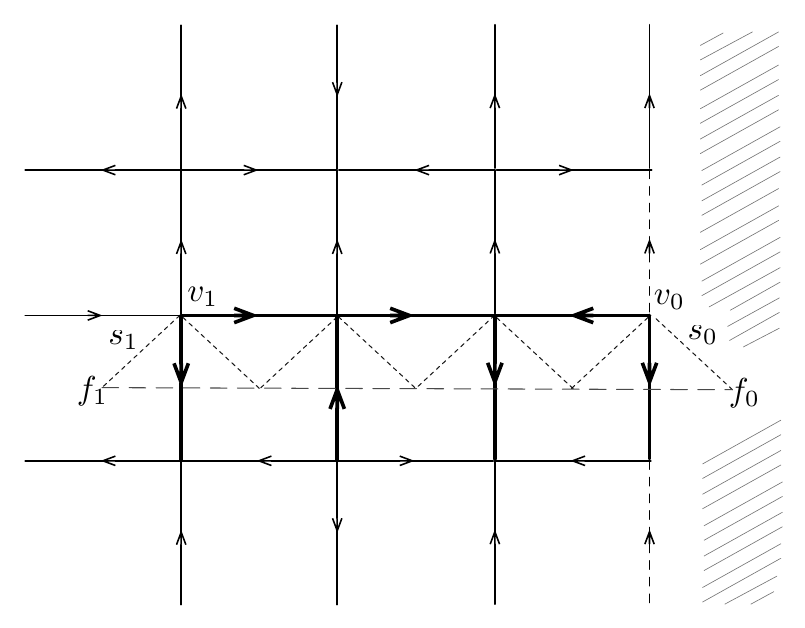}
}
\caption{A lattice with boundary in which every other boundary edge is marked by a dotted line. Here we fix a boundary site $s_0=(v_0, f_0)$ where $f_0$ is adjacent to a solid boundary edge.}
\label{fig:lattice-boundary-tilde}
\end{figure}

\subsection{A new boundary}

Consider a lattice on a half-plane as before. The vertex and face operators for internal sites remain unchanged. Also, the state of the boundary edges lives in the space $\mathbb{C}K$ which as before can be captured by adding the projections $B_s^K$ (for boundary sites $s$) to the Hamiltonian. However, we need to change the vertex operators at the boundary.

We assume that at the boundary, every other edge is marked by a dotted line as in Figure~\ref{fig:lattice-boundary-tilde}. Thus, for every boundary site there are three adjacent edges: (1) an internal edge, (2) a boundary \emph{solid} edge, (3) and a boundary \emph{dotted} edge. Then for every $k\in K$ and boundary site $s$ define the vertex operator $\widetilde{A}_s^{k}$ as in Figure~\ref{fig:a-tilde}. The equality $\widetilde{A}_s^k \widetilde{A}_s^l = \widetilde{A}_s^{kl}$ can be proved using \eqref{eq:cocycle}.\footnote{This is the equation that motivates us to have two types of boundary edges; if all the boundary edges were solid, we had $\widetilde{A}_s^k \widetilde{A}_s^l = \varphi(k, l)^2 \widetilde{A}_s^{kl}$.} Furthermore, using Lemma~\ref{lem:2-cocycle} we have $(\widetilde{A}_s^{k})^{\dagger} = \widetilde{A}_s^{k^{-1}}$. Then 
\begin{align*}
\widetilde{A}_s^K = \frac{1}{\vert K\vert}\sum_{k\in K} \widetilde{A}_{s}^k 
\end{align*}
is a projection.  
Now define the Hamiltonian 
\begin{align}\label{eq:hamiltonian-g-k-tilde}
\widetilde{H}_{G, K} =  - \sum_v A_v - \sum_f B_f -\sum_{s} (\widetilde{A}_s^K + B_s^K),
\end{align}
where the summations run over internal vertices and faces $v$, $f$ and boundary sites $s$. The ground state of this Hamiltonian is denoted by $\vert \widetilde{\psi}_K\rangle$. The bulk excitations as before, are created by applying ribbon operators and are labeled by irreps of $D(G)$. Condensations, however, are different because the boundary terms have been changed.

\begin{figure}
\centering{
\includegraphics[width=3.5 in]{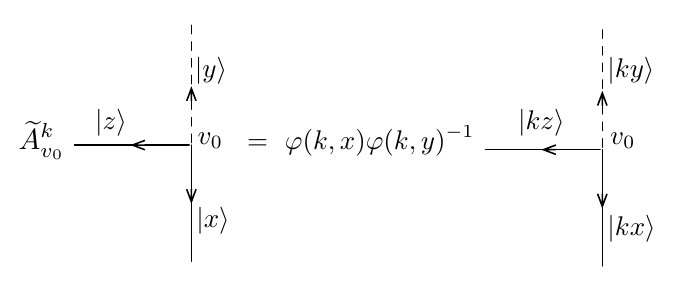}
}
\caption{Definition of $\widetilde{A}_{v_0}^k$ for $k\in K$ and boundary vertex $v_0$. Here we assume that $x, y$ are in $K$.}
\label{fig:a-tilde}
\end{figure}

\subsection{Condensations II}

In this section we fix a ribbon $\xi$ that connects a boundary site $s_0$ to an internal site $s_1$. Then the same as before we consider the subalgebra of ribbon operators that commute with $\widetilde{A}_{s_0}^K$ and $B_{s_0}^K$, and compute the corresponding representation of $D(G)$ in order to find the excitations that get condensed at $s_0$. Nevertheless, because we have changed the definition of the boundary terms, the ribbon operators when acting on boundary edges, must also be modified. 

From now on assume that the boundary edge corresponding to site $s_0$ is solid (see Figure~\ref{fig:lattice-boundary-tilde}). Then for $k\in K$ and $g\in G$ define the ribbon operator $\widetilde{F}_{\xi}^{k, g}$ as in Figure~\ref{fig:ribbon-tilde}. Observe that 
\begin{align}\label{eq:multi-f-tilde}
\widetilde{F}_{\xi}^{k, g} \widetilde{F}_{\xi}^{k', g'} & = \delta_{g, g'}  \varphi(k, k') \widetilde{F}_{\xi}^{kk', g},\\
(\widetilde{F}_{\xi}^{k, g})^{\dagger} &=\widetilde{F}_{\xi}^{k^{-1}, g}.
\end{align}
$\widetilde{F}_{\xi}^{k, g}$ commutes with $A_t$ and $B_t$ for every internal site $t\neq s_1$, and the commutation relations with $A_{s_1}^{h}$ and $B_{s_1}^h$ are as before
\begin{align}
A_{s_1}^h \widetilde{F}_{\xi}^{k,g} &= \widetilde{F}_{\xi}^{k, gh^{-1}} A_{s_1}^h \label{eq:com-1-a-tilde},\\
B_{s_1}^{h} \widetilde{F}_{\xi}^{k,g} &= \widetilde{F}_{\xi}^{k,g} B_{s_1}^{g^{-1}k^{-1}gh}, \label{eq:com-1-b-tilde}
\end{align}
and for $l\in K$
\begin{align}\label{eq:com-0-b-tilde}
B_{s_0}^l \widetilde{F}_{\xi}^{k, g} = \widetilde{F}_{\xi}^{k, g} B_{s_0}^{lk}.
\end{align}
However, we have 
\begin{align}\label{eq:com-0-a-tilde}
\widetilde{A}_{s_0}^l \widetilde{F}_{\xi}^{k, g} = \varphi(lk, l^{-1}) \varphi(l, k) \widetilde{F}_{\xi}^{lkl^{-1}, lg} \widetilde{A}_{s_0}^l.
\end{align}

\begin{remark} It is shown in~\cite{bombin} that the ribbon operators are indeed certain extensions of the local operators. Considering the extra phases in the definition of $\widetilde{A}_{s_0}^k$, the same extension gives the definition of $\widetilde{F}_{\xi}^{k,g}$, which comparing to $F_{\xi}^{k, g}$ has an extra phase.
\end{remark}

\begin{remark}
We defined the operators $\widetilde{F}_{\xi}^{k, g}$ only when $k$ belongs to $K$. This is because we are interested in ribbon operators  which commute with $B_{s_0}^K$. According to \eqref{eq:com-0-b-tilde}, this condition automatically enforces $k$ to be in $K$.
\end{remark}

\begin{figure}
\centering{
\includegraphics[width=4.8 in]{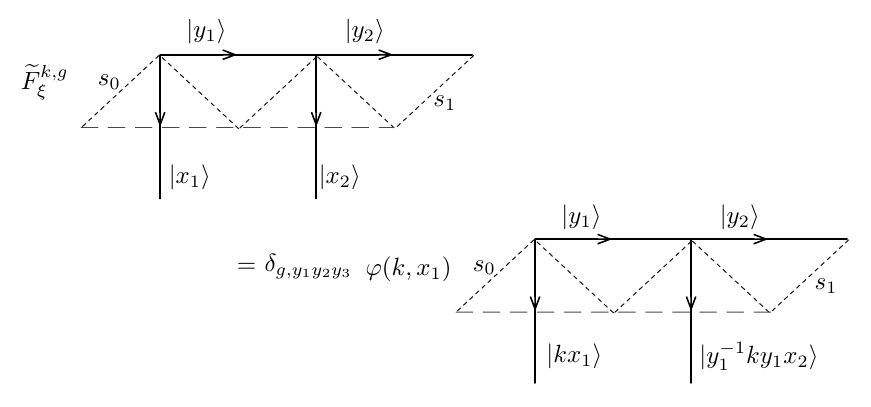}
}
\caption{Definition of $\widetilde{F}_{\xi}^{k, g}$ for a ribbon which connects a boundary site $s_0$ with a corresponding solid edge, to an internal site $s_1$. Here we assume that $k, x_1$ are in $K$.}
\label{fig:ribbon-tilde}
\end{figure}

Let $\widetilde{\mathcal{F}}_{\xi}$ be the algebra generated by $\widetilde{F}_{\xi}^{k, g}$, $k\in K$, $g\in G$, and 
\begin{align*}
\widetilde{\mathcal{C}}_{\xi} =\{T \in \widetilde{\mathcal{F}}_{\xi}:\, [T, \widetilde{A}_{s_0}^K ]=[ T, B_{s_0}^K]=0\}.
\end{align*} 
For every $k\in K$ and $g\in G$ define
\begin{align}\label{eq:def-t-tilde}
\widetilde{T}^{k, g} = \sum_{l\in K} \varphi(l, k) \varphi(lk, l^{-1}) \widetilde{F}_{\xi}^{lkl^{-1}, lg^{-1}}.
\end{align}

\begin{proposition}\label{prop:span-tilde}
The algebra $\widetilde{\mathcal{C}}_{\xi}$ is spanned by the operators $\widetilde{T}^{k,g}$, $k\in K$, $g\in G$.
\end{proposition}

\begin{proof}
$[B_{s_0}^K, \widetilde{T}^{k, g}]=0$ is easy. For $m\in K$ we have
\begin{align*}
\widetilde{A}_{s_0}^m \widetilde{T}^{k,g} & = \sum_{l\in K}  \varphi(l, k) \varphi(lk, l^{-1})\, \widetilde{A}_{s_0}^m \widetilde{F}_{\xi}^{lkl^{-1}, lg^{-1}} \\
&= \sum_{l\in K}  \varphi(l, k) \varphi(lk, l^{-1}) \varphi(m, lkl^{-1}) \varphi(mlkl^{-1}, m^{-1})\, \widetilde{F}_{\xi}^{mlkl^{-1}m^{-1}, mlg^{-1}} \widetilde{A}_{s_0}^m\\
&= \sum_{l\in K}  \varphi(m^{-1}l, k) \varphi(m^{-1}lk, l^{-1}m)  \varphi(m, m^{-1}l k l^{-1}m ) \varphi(l k l^{-1}m, m^{-1})\, \widetilde{F}_{\xi}^{lkl^{-1}, lg^{-1}} \widetilde{A}_{s_0}^m.
\end{align*}
So to prove $[\widetilde{A}_{s_0}^K, \widetilde{T}^{k, g}]=0$ it is sufficient to show that 
$$ \varphi(m^{-1}l, k) \varphi(m^{-1}lk, l^{-1}m)  \varphi(m, m^{-1}l k l^{-1}m ) \varphi(l k l^{-1}m, m^{-1})= \varphi(l, k) \varphi(lk, l^{-1}).$$
By three applications of \eqref{eq:cocycle} and equation 4 of Lemma \ref{lem:2-cocycle} we have
\begin{align*}
& \quad \,\, \varphi(m^{-1}l, k) \varphi(m^{-1}lk, l^{-1}m)  \varphi(m, m^{-1}l k l^{-1}m ) \varphi(l k l^{-1}m, m^{-1})\\
& = \varphi(m^{-1}l, k) \varphi(m, m^{-1}lk)  \varphi(lk, l^{-1} m ) \varphi(l k l^{-1}m, m^{-1})\\
& =  \varphi(m^{-1}l, k) \varphi(m, m^{-1}lk)  \varphi(l^{-1}m, m^{-1}) \varphi(l k , l^{-1})\\ 
& = \varphi(l, k) \varphi(m, m^{-1}l)  \varphi(l^{-1}m, m^{-1}) \varphi(l k , l^{-1})\\
& =  \varphi(l, k)  \varphi(l k , l^{-1}). 
\end{align*}
As a result, $\widetilde{T}^{k, g} \in \widetilde{\mathcal{C}}_{\xi}$. To see that these operators span $\widetilde{\mathcal{C}}_{\xi}$, observe that if $T = \sum c_{k, g} \widetilde{F}_{\xi}^{k, g}$ is in $\widetilde{\mathcal{C}}_{\xi}$, by $[T, B_{s_0}^K] =0 $, $c_{k, g} =0$ for every $k\notin K$. Moreover, by $[T, \widetilde{A}_{s_0}^K] =0$, $c_{mkm^{-1}, lg}$, for every $m\in K$, is uniquely determined in terms of $c_{k, g}$.

\end{proof}

\begin{proposition} \label{prop:rel-tilde}
The following equations hold.
\begin{enumerate}
\item $\widetilde{T}^{k, gm} = \varphi(m, k) \varphi(mk, m^{-1}) \, \widetilde{T}^{mkm^{-1}, g}$, for every $m\in K$.
\item $\widetilde{T}^{k,g}\widetilde{T}^{k',g'}=0$, if $gK\neq g'K$.
\item $\widetilde{T}^{k,g}\widetilde{T}^{k',g} = \varphi(k, k') \widetilde{T}^{kk', g}$.
\item $(\widetilde{T}^{k, g})^{\dagger} = \widetilde{T}^{k^{-1}, g}$.
\end{enumerate}
\end{proposition}

\begin{proof}
\begin{align*}
\widetilde{T}^{k, gm} & = \sum_{l\in K} \varphi(l, k) \varphi(lk, l^{-1}) \,\widetilde{F}_{\xi}^{lkl^{-1}, lm^{-1}g^{-1}}\\ 
&= \sum_{l\in K} \varphi(lm, k) \varphi(lmk, m^{-1}l^{-1}) \, \widetilde{F}_{\xi}^{lmkm^{-1}l^{-1}, lg^{-1}}.
\end{align*}
So for the first equation it is sufficient to show 
$$\varphi(lm, k) \varphi(lmk, m^{-1}l^{-1})= \varphi(m, k) \varphi(mk, m^{-1}) \varphi(l, mkm^{-1}) \varphi(lmkm^{-1}, l^{-1}).$$
Using Lemma \ref{lem:2-cocycle} we have
\begin{align*}
\varphi(lm, k) \varphi(lmk, m^{-1}l^{-1}) & = \varphi(lm, k) \varphi(l, m) \varphi(m^{-1}, l^{-1}) \varphi(lmk, m^{-1}l^{-1}) \\
& = \varphi(l, mk) \varphi(m, k) \varphi(lmk, m^{-1}) \varphi(lmkm^{-1} , l^{-1}) \\
& = \varphi(m, k) \varphi(lmkm^{-1}, l^{-1}) \varphi(l, mkm^{-1}) \varphi(mk, m^{-1}).
\end{align*}

The second equation is an easy consequence of \eqref{eq:multi-f-tilde}. For the third equation we have
\begin{align*}
\widetilde{T}^{k,g}\widetilde{T}^{k',g} & = \sum_{l\in K} \varphi(l, k) \varphi(lk, l^{-1}) \widetilde{F}_{\xi}^{lkl^{-1}, lg^{-1}} \sum_{l'\in K} \varphi(l', k') \varphi(l'k'l'^{-1}) \widetilde{F}_{\xi}^{l'k'l'^{-1}, l'g^{-1}}\\
&= \sum_{l\in K}  \varphi(l, k) \varphi(lk, l^{-1})  \varphi(l, k') \varphi(lk', l^{-1}) \varphi(lkl^{-1}, lk'l^{-1}) \widetilde{F}_{\xi}^{lkk'l^{-1}, lg^{-1}}.
\end{align*}
So we need to show that 
$$ \varphi(l, k) \varphi(lk, l^{-1})  \varphi(l, k') \varphi(lk', l^{-1}) \varphi(lkl^{-1}, lk'l^{-1}) = \varphi(k, k') \varphi(l, kk') \varphi(lkk', l^{-1}),$$
which can be proved by
\begin{align*}
& \quad \,\, \varphi(l, k) \varphi(lk, l^{-1})  \varphi(l, k') \varphi(lk', l^{-1}) \varphi(lkl^{-1}, lk'l^{-1}) \\
& = \varphi(l, k) \varphi(l, k') \varphi(lk', l^{-1}) \varphi(lk, k'l^{-1}) \varphi(l^{-1}, lk'l^{-1})\\
&= \varphi(l, k') \varphi(lk', l^{-1}) \varphi(l^{-1}, lk'l^{-1}) \varphi(l, kk'l^{-1}) \varphi(k, k'l^{-1})\\
& = \varphi(l, k'l^{-1}) \varphi(k', l^{-1}) \varphi(l^{-1}, lk'l^{-1}) \varphi(l, kk'l^{-1}) \varphi(k, k'l^{-1})\\
& = \varphi(l, k'l^{-1}) \varphi(l^{-1}, lk'l^{-1}) \varphi(l, kk'l^{-1}) \varphi(kk', l^{-1}) \varphi(k, k')\\
& = \varphi(l^{-1}, l ) \varphi(e, k'l^{-1}) \varphi(l, kk'l^{-1}) \varphi(kk', l^{-1}) \varphi(k, k')\\
& =  \varphi(k, k') \varphi(l, kk') \varphi(lkk', l^{-1}).
\end{align*}

For the last equation we have
\begin{align*}
(\widetilde{T}^{k, g})^{\dagger} & = \sum_{l\in K} \varphi(l, k)^{\ast} \varphi(lk, l^{-1 })^{\ast} \left( \widetilde{F}_{\xi}^{lkl^{-1}, lg^{-1}} \right)^{\dagger}\\
& = \sum_{l\in K} \varphi(k^{-1}, l^{-1}) \varphi(l, k^{-1}l^{-1}) \widetilde{F}_{\xi}^{lk^{-1}l^{-1}, lg^{-1}}\\
& = \sum_{l\in K} \varphi(l, k^{-1}) \varphi(lk^{-1}, l^{-1}) \widetilde{F}_{\xi}^{lk^{-1}l^{-1}, lg^{-1}}\\
& = \widetilde{T}^{k^{-1}, g}.
\end{align*}
\end{proof}

This proposition gives a full characterization of the algebra $\widetilde{\mathcal{C}}_{\xi}$. The next step is to compute the induced representation of $D(G)$. Let $\vert \widetilde{\psi}_K\rangle$ be the ground state of $\widetilde{H}_{G, K}$. Define $\vert \widetilde{\psi}_K^{k, g} \rangle = \widetilde{T}^{k, g} \vert \widetilde{\psi}_K \rangle $, and let $\mathcal{A}(K, \varphi)$ be the span of these vectors. Using \eqref{eq:com-1-a-tilde} and \eqref{eq:com-1-b-tilde} the representation $\mathcal{A}(K, \varphi)$ of $D(G)$ is given by 
\begin{align*}
A_{s_1}^{h} \vert \widetilde{\psi}_K^{k, g} \rangle & = \vert \widetilde{\psi}_K^{k, hg} \rangle, \\
B_{s_1}^h \vert \widetilde{\psi}_K^{k, g} \rangle & = \delta_{h, gkg^{-1}} \vert \widetilde{\psi}_K^{k, g} \rangle.
\end{align*}
To compute the character of this representation we need the following lemma.

\begin{lem} \label{lem:inner-tilde} The following equations hold for every $k, k'\in K$ and $g, g'\in G$.
\begin{enumerate}
\item $\langle \widetilde{\psi}_K \vert  \widetilde{F}_{\xi}^{k, g} \vert \widetilde{\psi}_K \rangle = \frac{1}{\vert G\vert} \delta_{k, e}$.
\item $\langle \widetilde{\psi}_K^{k, g} \vert \widetilde{\psi}_K^{k', g'} \rangle = 0 $ if $gK \neq g'K$.
\item $\langle \widetilde{\psi}_K^{k, g} \vert \widetilde{\psi}_K^{k', g} \rangle =  \frac{\vert K\vert}{\vert G\vert} \delta_{k, k'} $.
\end{enumerate}
\end{lem}

\begin{proof} The proof follows from similar steps as in the proof of Lemma \ref{lem:inner-product} and \eqref{eq:in-pro-2}.

\end{proof}

The following theorem gives a generalization of~\eqref{eq:char-phi=1}.

\begin{thm}\label{thm:a(k,varphi)} \cite{davydov} The character of the representation $\mathcal{A}(K, \varphi) $ of $D(G)$ is given by
\begin{align}\label{eq:character}
\chi_{\mathcal{A}(K, \varphi)}(gh^{\ast}) = \frac{1}{\v K\v}\, \delta_{gh, hg} \sum_{x\in G} \delta_{xgx^{-1}\in K} \delta_{xhx^{-1}\in K} \, \varphi(xgx^{-1}\v xhx^{-1}),
\end{align}
where $\varphi(k\v l) = \varphi(k, l) \varphi(klk^{-1}, k)^{-1} $ which equals $\varphi(k, l)\varphi(kl, k^{-1})$ if $\varphi$ satisfies Lemma~\ref{lem:2-cocycle}.
\end{thm}

\begin{proof}
Let $r= \vert G\vert / \vert K\vert$ and assume that $G = g_1 K \cup \dots \cup g_r K$. Then by Proposition \ref{prop:rel-tilde} and Lemma \ref{lem:inner-tilde}, $\{\sqrt{r} \vert \widetilde{\psi}_K^{k, g_i}\rangle :\, k\in K, i=1, \dots, r\}$ is an orthonormal basis for $\mathcal{A}(K, \varphi)$. Fix $g, h\in G$ and let $1\leq \epsilon(i) \leq r$ and $k_i\in K$ such that $gg_i = g_{\epsilon(i)} k_i$. We have 
\begin{align*}
\chi_{\mathcal{A}(K, \varphi)} (gh^{\ast}) & = r \sum_{k\in K} \sum_{i=1}^r \langle \widetilde{\psi}_K^{k, g_i} \vert gh^{\ast} \vert \widetilde{\psi}_K^{k, g_i} \rangle \\
& = r \sum_{k\in K} \sum_{i=1}^r  \delta_{h, g_i kg_i^{-1}} \langle \widetilde{\psi}_K^{k, g_i} \vert \widetilde{\psi}_K^{k, gg_i} \rangle\\
& = r \sum_{k\in K} \sum_{i=1}^r  \delta_{h, g_i kg_i^{-1}} \langle \widetilde{\psi}_K^{k, g_i} \vert \widetilde{\psi}_K^{k, g_{\epsilon(i)k_i}} \rangle \\
& = r \sum_{k\in K} \sum_{i=1}^r  \delta_{h, g_i kg_i^{-1}} \varphi(k_i\vert k)  \langle \widetilde{\psi}_K^{k, g_i} \vert \widetilde{\psi}_K^{k_ikk_i^{-1}, g_{\epsilon(i)}} \rangle \\
& = r \sum_{k\in K} \sum_{i=1}^r  \delta_{h, g_i kg_i^{-1}} \varphi(k_i\vert k)  \delta_{i, \epsilon(i)} \delta_{k, k_ikk_i^{-1}} \frac{\vert K\vert}{\vert G\vert} \\
& = \sum_{i=1}^{r} \delta_{g_i^{-1} gg_i \in K} \delta_{g_i^{-1} h g_i \in K}  \delta_{g_i^{-1} hg g_i ,\, g_i^{-1} g h g_i} \varphi(g_i^{-1}gg_i\vert g_i^{-1}hg_i)  \\
& = \frac{1}{\vert K\vert} \delta_{gh, hg} \sum_{x\in G} \delta_{xgx^{-1} \in K} \delta_{xhx^{-1} \in K} \varphi(xgx^{-1}\vert xhx^{-1}). 
\end{align*}

\end{proof}

Now to compute the condensed anyons we just need to decompose the character of the representation $\mathcal{A}(K, \varphi)$ into irreducible ones. An excitation $(\overline{a}, \pi)$ gets condensed at the boundary if $\chi_{(\overline{a}, \pi)}$ appears in this decomposition.

\begin{remark}
Note that if $\varphi'$ is another $2$-cocycle equivalent to $\varphi$ then $\varphi'(k\vert l) = \varphi(k \vert l)$ for every $k,l$ where $kl=lk$. Therefore, Theorem~\ref{thm:a(k,varphi)} still holds even if $\varphi$ does not satisfy Lemma~\ref{lem:2-cocycle}.
\end{remark}
 
\begin{remark}
$\mathcal{A}(K, \varphi)$ gives a representation of $D(G)$, so it is an object of $\mathcal{Z}(G)$. On the other hand, by Proposition \ref{prop:rel-tilde}, an algebra structure is defined on $\mathcal{A}(K, \varphi)$. Davydov has shown that these algebras are indeed all \emph{maximal indecomposable separable commutative algebras} of $\mathcal{Z}(G)$~\cite{davydov}.
\end{remark}

\section{A domain wall between two phases}\label{sec:domain-wall}

We now consider two phases corresponding to two groups $G$ and $G'$ and a domain wall between them, and want to study tunneling of anyons from the $G$-phase to the $G'$-phase. More precisely, we consider a planar lattice divided into two parts by a defect line; we associate the right half-plane to the quantum double model corresponding to group $G$, and the left half-plane to group $G'$; the terms of the Hamiltonian on the domain wall then depend on a subgroup $U\subseteq G\times G'$ and a $2$-cocycle $\varphi\in H^2(U, \mathbb{C}^{\times})$.  To understand these terms we can simply use the folding idea, i.e., we fold the plane through the defect line. Then we have one half-plane with a boundary. In this case, there are two vectors living on each edge of the lattice: one corresponding to $G$ and the other to $G'$. In other words, the Hilbert space associated with each edge of the lattice is $\mathbb{C}G\otimes \mathbb{C}G' \simeq \mathbb{C}(G\times G')$. Moreover, the boundary is parametrized by $U$ and the whole Hamiltonian would be equal to $\widetilde{H}_{G\times G', U}$ given by~\eqref{eq:hamiltonian-g-k-tilde}. 

Again, the bulk excitations correspond to irreps of $D(G\times G')$ which are basically pairs of irreps of $D(G)$ and $D(G')$ and can be interpreted as two anyons, one belonging to the $G$-phase and the other to the $G'$-phase. More formally, the bulk excitations correspond to $\mathcal{Z}(G\times G')$ which is equivalent to  $\mathcal{Z}(G)\boxtimes \mathcal{Z}(G')$ \cite{davydov}.  

Here there is a technical point. After folding, we change the orientation of the right half-plane, and then the braiding operator $C_{Y, Y'}$ (see Section~\ref{sec:z(g)}) will change to $C_{Y', Y}^{-1}$. That is why the excitation $X\boxtimes Y$ of $\widetilde{H}_{G\times G', U}$ corresponds to the pair of anyons $(X, Y^{op})$ in the unfolded plane, where if $Y=(\overline{y}, \pi)$ then $Y^{op}=(\overline{y}, \pi^{\ast})$ (see~\cite{davydov} for the definition of the opposite category). Therefore, in the unfolded plane, anyons on the right hand side indeed live in the category $\mathcal{Z}(G)^{op}$. However, as we will show later in this section $\mathcal{Z}(G)^{op}$ and $ \mathcal{Z}(G)$ are equivalent categories.

Now assume that an anyon $X$ in the $G$-phase, \emph{without creating any excitation on the domain wall}, tunnels to the $G'$-phase and becomes $Y$. This process in the folded half-plane is equivalent to creating  $X\boxtimes Y^{op}$ and moving it to the boundary so that it disappears, i.e., condensation of $X\boxtimes Y^{op}$. Since we know a classification of condensations of $\widetilde{H}_{G\times G', U}$ we can characterize tunnelings as well.

Let us give an example to clarify our framework. 
Assume that $G'=G$ and let $U=\Delta(G)=\{(g, g):\, g\in G\}$ and $\varphi$ be the trivial cocycle ($\varphi\equiv 1$). To find anyons in this condensation we use Theorem \ref{thm:a(k,varphi)}.
For $g=(g_1, g_2)$ and $h=(h_1, h_2)$ in $G\times G$ we have
\begin{align}
\chi_{\mathcal{A}(\Delta(G), 1)}(gh^{\ast})& = \frac{1}{\v \Delta(G)\v}\, \delta_{gh,hg}\sum_{x, y\in G} \delta_{xg_1x^{-1},yg_2 y^{-1}}\, \delta_{xh_1x^{-1}, yh_2 y^{-1}}\\
& = \delta_{gh,hg}\, \delta_{g_1h_1^{\ast}\overset{G}\sim g_2h_2^{\ast}} \,\v Z(g_1) \cap Z(h_1) \v,
\end{align}
where by $g_1h_1^{\ast}\overset{G}\sim g_2h_2^{\ast}$ we mean that there exists $l\in G$ such that $l(g_1h_1^{\ast})l^{-1} =g_2h_2^{\ast}$.
Now define
\begin{align}
\Phi(gh^{\ast}) = \sum_{(\overline{x}, \rho)} \chi_{(\overline{x}, \rho)\boxtimes (\overline{x}, \rho^{\ast})}\, (gh^{\ast}),
\end{align}
where the summation runs over all irreducible representations $(\overline{x}, \rho)$ of $D(G)$. We have
\begin{align}
\Phi(gh^{\ast}) & = \sum_{(\overline{x}, \rho)}\, \chi_{(\overline{x}, \rho)}(g_1h_1^{\ast})\, \chi_{(\overline{x}, \rho^{\ast})} (g_2h_2^{\ast}) \label{eq:phi1} \\
& = \delta_{g_1h_1 , h_1g_1} \delta_{g_2h_2 , h_2g_2} \delta_{h_1\overset{G}\sim h_2} \sum_{(\overline{h_1} , \rho)}  \tr_{\rho}(g_1)\, \tr_{\rho^{\ast}}(k_{h_2}^{-1} g_2 k_{h_2})\\
& = \delta_{gh , hg}  \delta_{h_1\overset{G}\sim h_2}
\delta_{g_1\overset{Z(h_1)}\sim  (k_{h_2}^{-1} g_2 k_{h_2}) } \v Z_{Z(h_1)}( g_1 ) \v \\
& = \delta_{gh , hg}  \delta_{g_1h_1^{\ast}\overset{G}\sim g_2h_2^{\ast}} \v Z( g_1 ) \cap Z(h_1) \v \label{eq:phi4},
\end{align}
where in the third line we use the orthogonality relations in the character table of $Z(h_1)$. As a result, $\chi_{\mathcal{A}(\Delta(G), 1)} = \Phi$, or equivalently anyons of the form $X\boxtimes X^{op}$ get condensed.\footnote{This example indeed show that the map $X\mapsto X^{op}$ gives the equivalence between two categories $\mathcal{Z}(G)$ and $\mathcal{Z}(G)^{op}$.}

\subsection{A non-trivial auto-equivalence of $\mathcal{Z}(\group)$}\label{sec:3}

The tunneling process may give an equivalence between two phases $G$ and $G'$. Suppose that condensations corresponding to $\mathcal{A}(U, \varphi)$ ($U\subseteq G\times G'$), is described by the character $\chi_{\mathcal{A}(U, \varphi)} = \sum_i \chi_{(X_i \boxtimes Y_i)}$. Then $X_i$ after tunneling, without creating any excitation at the domain wall, is changed to $Y_i$. On the other hand, fusions and braidings are invariant under tunneling. Therefore, if $X_i$'s and $Y_i$'s are all simple objects of $\mathcal{Z}(G)$ and $\mathcal{Z}(G')$ respectively, then $X_i\mapsto Y_i^{op}$ gives an equivalency between anyons of the $G$-phase and $G'$-phase.\footnote{$X_i\mapsto Y_i$ gives the equivalence $\mathcal{Z}(G)\simeq \mathcal{Z}(G')^{op}$ which by combining with the equivalence $\mathcal{Z}(G')^{op}\simeq \mathcal{Z}(G')$ we find that $\mathcal{Z}(G)\simeq \mathcal{Z}(G')$ is given by $X_i \mapsto Y_i^{op}$.} Using this idea we show a non-trivial symmetry in $\mathcal{Z}(\group)$.

Let $\mathbf{F}_q$ be the finite field with $q$ elements and denote its additive and multiplicative groups by $\mathbf{F}_q^+$ and $\mathbf{F}_q^{\times}$ respectively. Then the semidirect product of these groups is defined as follows. We represent elements of $\group$ by $(a, \alpha)$ where $a\in \mathbf{F}_q^{+}$ and $\alpha\in \mathbf{F}_q^{\times}$, and define $(a, \alpha)(a', \alpha') = (a+ \alpha \times a', \alpha \times \alpha')$ which by abuse of notation is denoted by $(a+ \alpha a', \alpha\alpha')$. The identity element of this group is $e=(0, 1)$ and the inverse of $(a, \alpha)$ is equal to $(a, \alpha)^{-1}= (-\alpha^{-1}a, \alpha^{-1})$.

We will use the following properties of $\group$. The conjugacy class of $(a, \alpha)$ is $\overline{(a, \alpha)}=\{(b, \alpha):\, b\in \mathbf{F}_q^{+}\}$ if $\alpha\neq 1$, and $\overline{(1, 1)}=\{(b, 1):\, b\in \mathbf{F}_q^{+}, b\neq 0\}$.
$K= \{ (a, 1):\, a\in \mathbf{F}_q^{+}  \}$ is a normal subgroup of $\group$ isomorphic to $\mathbf{F}_q^{+}$, and $K=Z(1, 1)$. Also note that for every $(a, \alpha)$ where $\alpha\neq 1$, $\v Z(a, \alpha)\v =q-1$ and $K\cap Z(a, \alpha)=\{e\}$.

Consider a non-trivial irrep of $K$, and let $\pi$ be the corresponding induced representation on $\group$. Then $\tr_{\pi} (e) = q-1$, $\tr_{\pi} (1, 1) =-1$, and $\tr_{\pi} (a, \alpha)=0$ if $\alpha\neq 1$. Since $\sum_{(a, \alpha)} \v \tr_{\pi} (a, \alpha)\v^2 = \v \group\v$, $\pi$ is an irreducible representation of $\group$.

\begin{thm}\label{thm:main}
There exists an auto-equivalence of $\mathcal{Z}(\group)$ whose corresponding permutation on simple objects is of the form $PJ$ where $J$ sends every object to its charge conjugation ($J: X \mapsto X^{\vee} $) and $P$ is the transposition of the chargeon $C=(e, \pi)$ and fluxion $F=(\overline{(1, 1)}, \mathbf{1})$.
\end{thm}

\begin{proof} Define 
$$U= \{ ((a_1, \alpha), (a_2, \alpha^{-1})):\, a_1,a_2\in \mathbf{F}_q^{+},\, \alpha\in \mathbf{F}_q^{\times} \}.$$ 
Let $p$ be the characteristic of $\mathbf{F}_q$ (so $q$ is a power of $p$), and let $\omega$ be a $p$-th root of unity ($\omega^p=1$). Additionally, assume that $\tr_p: \mathbf{F}_q \rightarrow \mathbf{F}_p$ is the trace function, i.e., $\tr_p(a)$ is equal to the trace of the $\mathbf{F}_p$-linear map $x\mapsto ax$. Now define $\varphi: U\times U\rightarrow \mathbb{C}^{\times}$ by
\begin{align}
\varphi( g,  h ) = \omega^{\tr_p( \alpha a_2b_1)},
\end{align}
where $g=((a_1, \alpha),( a_2, \alpha^{-1}))$ and $h=((b_1, \beta),( b_2, \beta^{-1}))$. $\varphi$ satisfies~\eqref{eq:cocycle}, and then $\varphi\in H^2(U, \mathbb{C}^{\times})$. 
In Appendix~\ref{app:a} it is shown that the character of the representation $\mathcal{A}(U, \varphi)$ is given by $\chi_{\mathcal{A}(U, \varphi)}(gh^{\ast}) =0$ if $g$ or $h$ is not in $U$, and 
\begin{align}
\chi_{\mathcal{A}(U, \varphi)}(gh^{\ast}) =
\begin{cases}
  \delta_{gh,hg}\, \delta_{g, h \in U}\, (q-1)  & \mbox{if } \alpha\neq 1 \mbox{ or } \beta\neq 1, \\
  \delta_{gh,hg}\, \delta_{g, h \in U}\, \left( \delta_{a_1 b_2 , a_2 b_1} (q-1) -  \delta_{a_1 b_2  \neq a_2 b_1}\right) & \mbox{if } \alpha = \beta =1,
\end{cases}
\end{align}
if $g=((a_1, \alpha),( a_2, \alpha^{-1}))$ and $h=((b_1, \beta),( b_2, \beta^{-1}))$ belong to $U$.
Furthermore, it is shown that $\chi_{\mathcal{A}(U, \varphi)} = \Psi - \Gamma$ where 
\begin{align}
\Psi(gh^{\ast}) = \sum_{(\overline{x}, \rho)} \chi_{(\overline{x}, \rho) \boxtimes (\overline{x^{-1}}, \rho)} (gh^{\ast}),
\end{align}
and $\Gamma = \chi_{C\boxtimes C} + \chi_{F\boxtimes F} - \chi_{C\boxtimes F} - \chi_{F\boxtimes C}$.
As a result, $\mathcal{A}(U, \varphi)$ gives an auto-equivalence of $\mathcal{Z}(\group)$ which transposes $C$ and $F$ and sends $(\overline{x}, \rho)\neq C,F$ to $(\overline{x^{-1}}, \rho)^{op} = (\overline{x}, \rho)^{\vee}$. 

\end{proof}

$q=2$, the simplest example of this theorem, gives the group $\mathbb{Z}_2$, and the corresponding auto-equivalence is described in Section~\ref{sec:intro}.

For $q=3$ the group $\group$ is isomorphic to $S_3$, and the chargeon and fluxion constructed in the proof, correspond to representations $C$ and $F$ described in Section~\ref{sec:s3}. Moreover, in $\mathcal{Z}(S_3)$ the charge conjugation of each particle is itself. Thus this auto-equivalence of $\mathcal{Z}(S_3)$ only transposes $C$ and $F$, which means that these two particles in $\mathcal{Z}(S_3)$ are indistinguishable.

In Appendix~\ref{app:b} we show that a group $G$ has a symmetry similar to that of $\group$ only if $G\simeq \near$ where $\mathbf{H}$ is a finite \emph{near-field}.

\subsection{When are $\mathcal{Z}(G)$ and $\mathcal{Z}(G')$ equivalent?}

The proof of Theorem~\ref{thm:main} is based on the fact that there exists $\mathcal{A}(U, \varphi)$ such that $\chi_{\mathcal{A}(U, \varphi)}  = \sum_{i} \chi_{X_i\boxtimes Y_i}$ gives a permutation between anyons of the two phases. So if we find the necessary and sufficient condition for the existence of such $U\subseteq G\times G'$ and $\varphi\in H^2(U, \mathbb{C}^{\times})$  we can answer the question of whether $\mathcal{Z}(G)$ and $\mathcal{Z}(G')$ are equivalent. This question was first answered by Naidu and Nikshych~\cite{naidu, lagrangian} based on the classification of Lagrangian subcategories of $\mathcal{Z}(G)$.  Here we state the necessary and sufficient conditions of Davydov which is more appropriate for us.

\begin{thm} \label{thm:automorphism} \cite{davydov} An equivalence between $\mathcal{Z}(G)$ and $\mathcal{Z}(G')$ corresponds to a subgroup $U\subseteq G\times G'$, and $\varphi \in H^{2}(U, \mathbb{C}^{\times})$ such that
\begin{enumerate}
\item the projections of $U$ onto the first and second components are equal to $G$ and $G'$ respectively, and
\item the restriction of $\varphi(\cdot \v \cdot) $ (defined in Theorem~\ref{thm:a(k,varphi)}) on $(U\cap (G\times \{e\}))\times (U\cap ( \{e\}\times G'))$ is non-degenerate.
\end{enumerate}
Moreover, if such a $U$ and $\varphi$ exist, the map of the corresponding equivalence on simple objects can be computed by decomposing $\chi_{\mathcal{A}(U, \varphi)}$ into irreducible characters of $D(G\times G')$; if ${X\boxtimes Y}$ appears in this decomposition, then the equivalence sends $X$ to $Y^{op}$.
\end{thm}

Observe that the subgroup $U$ and $2$-cocycle $\varphi$ defined in the proof of Theorem~\ref{thm:main} satisfy the conditions of this theorem.

The framework of tunneling can be considered for any $U$ and $\varphi$ and not necessarily those given by the above theorem. However, in general we obtain an equivalence between certain subcategories of $\mathcal{Z}(G)$ and $\mathcal{Z}(G')$ and not necessarily the whole categories (see Theorem 2.5.1 of~\cite{davydov}).

\section{Conclusion}\label{sec:conclusion}

In this paper we defined the quantum double model with boundaries and found the corresponding condensations. Our work is based on the characterization of algebras in $\mathcal{Z}(G)$. However, the algebras that we constructed are the maximal ones classified in~\cite{davydov}. Indeed, indecomposable separable commutative algebras of $\mathcal{Z}(G)$ are indexed by $\mathcal{A}(M, K, \varphi, \varepsilon)$ where $K\subseteq M$ are subgroups of $G$ and $K$ is normal in M, $\varphi\in H^2(K, \mathbb{C}^{\times})$, and $\varepsilon$ is some extension of $\varphi$ to $M\times K$. $\mathcal{A}(M, K, \varphi, \varepsilon)$ is maximal if $M=K$ and in this case $\varepsilon$ is uniquely determined in terms of $\varphi$. It is an interesting question whether we can define a boundary for the quantum double model so that the corresponding condensation is given by $\mathcal{A}(M, K, \varphi, \varepsilon)$. Since the model of~\cite{bombin} is defined based on two subgroups of $G$, a combination of the ideas of the current paper and~\cite{bombin} may answer this question. 

The condensations that we characterized are indeed single-quasiparticle excitations, and we know that such excitations do not exist in the usual quantum double model. In the extended quantum double model of~\cite{bombin}, however, single-quasiparticles are possible on surfaces with non-trivial topology. So it is interesting to see what happens to boundaries on surfaces beyond plane and sphere.  

Characterization of confinements as well as edge excitations in these models is another important problem. Classification of edge excitations will clarify domain wall excitations as well.

Our general framework for studying boundaries allows us to examine the known facts about the toric code with boundary for the non-abelian quantum double models. See~\cite{anomaly, zohar} for some results in this direction.

In the second part of the paper, by applying the folding idea we considered the problem of tunneling of an excitation from one phase to another one, and then explained the necessary and sufficient conditions on two groups $G$ and $G'$ such that $\mathcal{Z}(G)\simeq \mathcal{Z}(G')$. Based on this approach, we found some non-trivial auto-equivalence of $\mathcal{Z}(\group)$. Finding other such symmetries and their applications are also of interest. For example Bombin in~\cite{ising}, using the symmetry in the case of $G=\mathbb{Z}_2$ have realized Ising anyons from an abelian model. \\

\noindent{\bf Acknowledgements.} This paper would have never had this shape without several helpful discussions with Alexei Kitaev, so we gratefully acknowledge him. We are also thankful to Miguel A. Martin-Delgado for introducing his work on condensations in the Kitaev model, and Liang Kong, Chris Heunen, Alexei Davydov, and John Preskill for many clarifications.

\appendix

\vspace{.3in}
\begin{center}
{\Large \bf Appendix}
\end{center}

\section{Proof of Theorem~\ref{thm:main}}\label{app:a}

To compute $\chi_{\mathcal{A}(U, \varphi)}$ we use Theorem~\ref{thm:a(k,varphi)}. Since $U$ is a normal subgroup we have
\begin{align*}
\chi_{\mathcal{A}(U, \varphi)}(gh^{\ast}) = \frac{1}{\v U\v} \delta_{gh,hg}\, \delta_{g, h \in U}\, \sum_{k} \varphi(kgk^{-1}, khk^{-1}) \varphi( khk^{-1}, kgk^{-1} )^{-1}.
\end{align*}
Letting $g=((a_1, \alpha),( a_2, \alpha^{-1}))$, $h=((b_1, \beta),( b_2, \beta^{-1}))$ and $k=((x_1, \theta), (x_2, \lambda))$, we have
\begin{align*}
kgk^{-1} & =  ( (x_1+\theta a_1 - \alpha x_1, \alpha), (x_2+\lambda a_2 - \alpha^{-1} x_2 , \alpha^{-1} )  ),\\
khk^{-1} & =  ( (x_1+\theta b_1 - \beta x_1, \beta), (x_2+\lambda b_2 - \beta^{-1} x_2 , \beta^{-1} )  ).
\end{align*}
and thus
\begin{align*}
\varphi(kgk^{-1}, khk^{-1}) & = \omega^{\tr_p( ( \alpha x_2+\alpha \lambda a_2 -x_2  )( x_1+ \theta b_1 - \beta x_1  )    )},\\
\varphi( khk^{-1}, kgk^{-1} ) & = \omega^{\tr_p(  ( \beta x_2+\beta \lambda b_2 -x_2  )( x_1+ \theta a_1 - \alpha x_1  )  )}.
\end{align*}
Now observe that $gh=hg$ is equivalent to $b_1(\alpha -1) =a_1(\beta -1)$ and $\alpha(1-\beta )a_2 = \beta(1-\alpha) b_2$. So if $g$ and $h$ commute, $\varphi(kgk^{-1}, khk^{-1}) \varphi( khk^{-1}, kgk^{-1} )^{-1}$ is independent of $x_1, x_2$, and we have
\begin{align*}
\chi_{\mathcal{A}(U, \varphi)}(gh^{\ast}) = \frac{1}{\v U\v} \delta_{gh,hg}\, \delta_{g, h \in U}\, \sum_{x_1,x_2, \theta, \lambda} \omega^{\tr_p( \theta\lambda ( \alpha a_2 b_1 - \beta b_2 a_1)  )}.
\end{align*}
Therefore,
\begin{align}\label{eq:27}
\chi_{\mathcal{A}(U, \varphi)}(gh^{\ast}) =
\begin{cases}
  \delta_{gh,hg}\, \delta_{g, h \in U}\, (q-1)  & \mbox{if } \alpha b_1 a_2  = \beta b_2 a_1, \\
  - \delta_{gh,hg}\, \delta_{g, h \in U} & \mbox{if } \alpha b_1 a_2  \neq \beta b_2 a_1 .
\end{cases}
\end{align}
Note that if either $\alpha$ or $\beta$ is not equal to $1$, then $g, h\in U$ and $gh=hg$  imply $\alpha b_1 a_2  = \beta b_2 a_1$. Thus~\eqref{eq:27} can be simplified to
\begin{align*}
\chi_{\mathcal{A}(U, \varphi)}(gh^{\ast}) =
\begin{cases}
  \delta_{gh,hg}\, \delta_{g, h \in U}\, (q-1)  & \mbox{if } \alpha\neq 1 \mbox{ or } \beta\neq 1, \\
  \delta_{gh,hg}\, \delta_{g, h \in U}\, \left( \delta_{a_1 b_2 , a_2 b_1} (q-1) -  \delta_{a_1 b_2  \neq a_2 b_1}\right) & \mbox{if } \alpha = \beta =1 .
\end{cases}
\end{align*}

We now need to decompose $\chi_{\mathcal{A}(U, \varphi)}$ into irreducible characters. Let
\begin{align*}
\Psi(gh^{\ast}) = \sum_{(\overline{x}, \rho)} \chi_{(\overline{x}, \rho) \boxtimes (\overline{x^{-1}}, \rho)} (gh^{\ast}).
\end{align*}
By the same steps as in the computation of $\Phi(gh^{\ast})$ in~\eqref{eq:phi1}-\eqref{eq:phi4} we find that
\begin{align*}
\Psi(gh^{\ast}) =  \delta_{gh,hg}\, \delta_{ g_1h_1^{\ast}\sim g_2^{-1}(h_2^{-1})^{\ast}     }\, \v Z(g_1 )\cap Z( h_1 ) \v,
\end{align*}
where $g=(g_1, g_2)=( (a_1, \alpha), (a_2, \alpha'))$ and $h=(h_1, h_2)=((b_1, \beta), (b_2, \beta'))$.
Observe that if $gh=hg$ and either $\alpha\neq 1$ or $\beta\neq 1$, then $g_1h_1^{\ast}\sim g_2^{-1}(h_2^{-1})^{\ast} $ is equivalent to $g,h\in U$. This fact can be verified simply by writing these conditions in terms of $a_1, a_2, \alpha,$ etc. Also in this case $g,h\in U$ and $gh=hg$ imply $\v Z(g_1 )\cap Z( h_1 ) \v = q-1$. Moreover, if $\alpha=\beta=1$, then $ g_1h_1^{\ast}\sim g_2^{-1}(h_2^{-1})^{\ast}$ is equivalent to $a_1 b_2 = a_2 b_1$, $g_1\sim g_2^{-1}$, and $h_1\sim h_2^{-1}$. Therefore,
\begin{align*}
\Psi(gh^{\ast})= \begin{cases}
  \delta_{gh, hg}\, \delta_{g, h \in U}\, (q-1)  & \mbox{if } \alpha \neq 1 \mbox{ or }\beta\neq 1, \\
  \delta_{gh, hg}\, \delta_{g_1\sim g_2^{-1}} \delta_{h_1\sim h_2^{-1}} \delta_{a_1b_2,a_2b_1}\, \v Z(g_1 )\cap Z( h_1 ) \v & \mbox{if } \alpha=\beta=1 .
\end{cases}
\end{align*}

Define
$\Gamma = \chi_{C\boxtimes C} + \chi_{F\boxtimes F} - \chi_{C\boxtimes F} - \chi_{F\boxtimes C}$. Then
\begin{align*}
\Gamma(gh^{\ast}) & = \left( \chi_C(g_1h_1^{\ast}) -\chi_F(g_1h_1^{\ast}) \right) \left( \chi_C(g_2h_2^{\ast}) - \chi_F(g_2h_2^{\ast}) \right)\\
& = \delta_{gh,hg} \left(  \delta_{h_1,e}\, \tr_{\pi}(g_1) - \delta_{h_1\in \overline{(1, 1)}}  \right)\left(  \delta_{h_2,e}\, \tr_{\pi}(g_2) - \delta_{h_2\in \overline{(1, 1)}}    \right).
\end{align*}

If either $\alpha\neq 1$ or $\beta\neq 1$, then $\Gamma(gh^{\ast})=0$ and we have $\chi_{\mathcal{A}(U, \varphi)}(gh^{\ast})=\Psi(gh^{\ast})=\Psi(gh^{\ast})-\Gamma(gh^{\ast})$. Moreover, when $\alpha=\beta=1$ by considering a few cases one can verify that $\chi_{\mathcal{A}(U, \varphi)}(gh^{\ast}) = \Psi(gh^{\ast}) -\Gamma(gh^{\ast})$. For instance, if ($\alpha=\beta=1$ and) $a_1=0$ and $a_2\neq 0$ we have
\begin{align*}
\chi_{\mathcal{A}(U, \varphi)}(gh^{\ast}) & =\delta_{g_2h_2,h_2g_2} \delta_{\alpha'=\beta'=1} (\delta_{b_1=0}(q-1) - \delta_{b_1\neq 0})\\
& = \delta_{\alpha'=\beta'=1} (\delta_{b_1=0} (q-1) - \delta_{b_1\neq 0}),
\end{align*}
and
\begin{align*}
\Gamma(gh^{\ast}) &= \delta_{g_2h_2,h_2g_2}\delta_{\alpha'=\beta'=1} (\delta_{b_1=0}(q-1) - \delta_{b_1\neq 0})( -\delta_{b_2=0} - \delta_{b_2\neq 0})\\
&= \delta_{\alpha'=\beta'=1} (\delta_{b_1=0}(q-1) - \delta_{b_1\neq 0})(-1),
\end{align*}
and since $g_1=e$ is not conjugate with $g_2^{-1}\neq e$, $\Psi(gh^{\ast})=0$. Thus $\chi_{\mathcal{A}(U, \varphi)}(gh^{\ast})=\Psi(gh^{\ast})=\Psi(gh^{\ast})-\Gamma(gh^{\ast})$.

As a result, $\mathcal{A}(U, \varphi)$ corresponds to an auto-equivalence of $\mathcal{Z}(\group)$ which transposes $C$ and $F$ and sends $(\overline{x}, \rho)\neq C,F$ to $(\overline{x^{-1}}, \rho)^{op} = (\overline{x}, \rho)^{\vee}$.

\section{Chargeon-fluxion symmetry as a modular invariant}\label{app:b}

The corresponding $S$-matrix to $\mathcal{Z}(G)$ is defined in~\eqref{eq:s-matrix}. The $T$-matrix is a diagonal one that contains the \emph{twist numbers} of simple objects on the diagonal. For $\mathcal{Z}(G)$, $T$ is given by
\begin{align}\label{eq:t-matrix}
T_{(\overline{g}, \pi)(\overline{g}, \pi)} = T_{(\overline{g}, \pi)} = \frac{\tr_{\pi} (g)}{\tr_{\pi}(e)}.
\end{align}
The pair of matrices $(S, T)$ is called a \emph{modular data}, and a \emph{modular invariant} corresponding to $(S, T)$ is a matrix $M$ that commutes with both $S$ and $T$, and such that all entries of $M$ are non-negative integers and $M_{\mathbf{0}\mathbf{0}}=1$ ($\mathbf{0}$ is the trivial object).
Clearly, the permutation corresponding to an auto-equivalent of a modular tensor category commutes with both $S$ and $T$ and is a modular invariant. However, a modular invariant may not even be a permutation and then may not come from an auto-equivalence.

In this section we study permutation matrices which form a modular invariant of $\mathcal{Z}(G)$. In particular, we classify all groups $G$ for which there exists a modular invariant of the form $P$ or $PJ$, where $P$ is a transposition of a chargeon-fluxion pair. Note that $J$ always commutes with both $S$ and $T$ (it can easily be seen from the formulas~\eqref{eq:s-matrix} and~\eqref{eq:t-matrix} in the case of $\mathcal{Z}(G)$; for a proof in the general case see \cite{bakalov}). Thus, $PJ$ is a modular invariant if and only if $P$ is a modular invariant.

\subsection{Near-fields}\label{sec:4.1}

By the result of Section~\ref{sec:3}, all groups $\group$, defined in terms of a finite field, admit a transposition of a chargeon-fluxion pair as a modular invariant. Here we show that every group with a modular invariant of this form is isomorphic to $\near$ where $\mathbf{H}$ is a \emph{near-field}.

\begin{definition}
A set $\mathbf{H}$ with two binary operations $+$ and $\times$ is called a near-field if
\begin{enumerate}
\item $(\mathbf{H}, +)$ is an abelian group with the identity element $0$.

\item $0\times x = x\times 0 =0$ for every $x\in \mathbf{H}$.

\item $(\mathbf{H}\setminus 0, \times)$ is a group with the identity element $1$.

\item the multiplication is distributive from left with respect to the addition: $x\times (y + z) = x\times y + x\times z$. (Distributivity from right is not assumed.)
\end{enumerate}
\end{definition}

The class of all finite near-fields is completely known: there is a method for constructing finite near-fields due to Dickson \cite{dickson}, and it has been shown by Zassenhaus \cite{Zassenhaus} that all finite near-fields except precisely seven of them, are given by Dickson's construction.

For a near-field $\mathbf{H}$ one can consider an action of $\mathbf{H}^{\times}$ on $\mathbf{H}^+$ and define a group structure on $\near$ as follows. Elements of $\near$ are denoted by $(a, \alpha)$ where $a\in \mathbf{H}^{+}$ and $\alpha\in \mathbf{H}^{\times}$, and $(a, \alpha)(b, \beta) = (a+\alpha\times b , \alpha\times \beta)$. This multiplication turns $\near$ to a group with the identity element $e=(0, 1)$. (Note that to obtain a group we must define the action of $\mathbf{H}^{\times}$ on $\mathbf{H}^{+}$ by multiplication from \emph{left}, and multiplication from right does not work.) $K=\{(a, 1):  a\in \mathbf{H}^+ \}$ is a subgroup of $\near$ isomorphic to $\mathbf{H}^{+}$. On the other hand, it is easy to see that all elements of $K\setminus e$ are conjugate. Thus, $K$ is an abelian group all of whose elements, except identity, have the same order. As a result, the size of this group $\v K\v= \v \mathbf{H}\v =q$ is a power of a prime number, and $K\simeq \mathbf{H}^+ \simeq \mathbf{F}_q^{+}$.

We will also use the fact that the centralizer of the multiplicative group of every near-field (with more than $2$ elements) is non-trivial. This property can be verified by checking Dickson's near-fields as well as the other seven near-fields classified by Zassenhaus (see \cite{hall}).

\subsection{A group with a chargeon-fluxion symmetry is isomorphic to $\near$}

We now state the main result of this section.

\begin{thm} \label{thm:uniqueness}
Suppose that the permutation matrix $P$ corresponding to a transposition of a chargeon-fluxion pair forms a modular invariant for $\mathcal{Z}(G)$. Then $G\simeq \near$ where $\mathbf{H}$ is a near-field. Conversely, for every group $\near$ there exists such a modular invariant.
\end{thm}

\begin{proof}
We first show that there exists a chargeon-fluxion pair in $\mathcal{Z}(\near)$ that forms a modular invariant.

Consider a non-trivial representation of the abelian subgroup $K\subseteq \near$ (defined above), and denote its induced representation on $\near$ by $\pi$. Let $C=(e, \pi)$ and $F=(\overline{a}, \mathbf{1})$, where $a=(1,1)\in \near$. We claim that the permutation $P$ which exchanges $C$ and $F$ is a modular invariant.

By the definition of $\pi$, $\dim \pi=q-1$, $\tr_{\pi}(a)=-1$ and $\tr_{\pi}(h) = 0$ for every $h\notin K$. Then since $\sum_{g} \v\tr_{\pi}(g)\v^2 =q(q-1)$, $\pi$ is an irreducible representation. Also a dimension-counting argument shows that all other irreducible representations of $\near$ come from an irreducible representation of $\mathbf{H}^{\times} \simeq (\near)/K$, and then for every such representation $\mu$, $\tr_{\mu}(a)=\tr_{\mu}(e)=\dim \mu$.

$P$ commutes with $T$ because $T_{C}=T_{F}=1$. To prove $PS=SP$ we should show that $S_{CX}=S_{FX}$ for every irreducible representation $X\neq C, F$ of $D(\near)$, and $S_{CC}=S_{FF}$. This is a straightforward computation given the structure of $Z(a)=K$ and the irreducible representations of $\near$ described above.

Now consider an arbitrary group $G$, let $C=(e, \pi)$ be a chargeon and $F=(\overline{a}, \mathbf{1})$ a fluxion in $\mathcal{Z}(G)$, and assume that the permutation $P$ which interchanges $C$ and $F$ commutes with the corresponding $S$-matrix. Note that since by the Verlinde formula the fusion rules are computed in terms of $S$ and $PSP^{-1}=S$, the fusion rules are also symmetric with respect to $C$ and $F$. We prove $G\simeq \near$ in the following steps.\\

\noindent(a) $\pi\neq 1$ and $a \neq e$.\\

$\mathbf{0}=(e, 1)$ is the unique representation such that $\mathbf{0}\otimes X \simeq X$, so its fusion rules cannot be the same
as any other representation. Therefore, $C$ and $F$ are different from $\mathbf{0}$. $\Box$\\

\noindent(b) $\dim \pi = \v \overline{a}\v$.\\

Since $\mathbf{0}\neq C, F$ we have $S_{C\mathbf{0}}=S_{F\mathbf{0}}$. Then
\begin{align} \label{eq:dim}
\frac{\dim \pi}{\vert G\vert}=\frac{1}{\vert Z(a)\vert},
\end{align}
or equivalently $\dim \pi = \v \overline{a}\v$. $\Box$ \\

\noindent(c) $\{e\}\cup\overline{a}$ is a subgroup of $G$.\\

Let $X=(\overline{h},\mu)$ be a representation such that $\overline{h}$ is different from $\{e\}$ and $\overline{a}$. Since $C$ has a trivial magnetic flux, $C\otimes X$ is equivalent to the sum of representations whose magnetic flux is equal to $\overline{h}$. Thus, the magnetic flux of any representation in $F\otimes X$ should also be $\overline{h}$. This means that $\overline{a}\,\overline{h}=\overline{h}$ for any $h\notin \{e\}\cup\overline{a}$. As a result,  $\overline{a}\, \overline{a}\subseteq \{e\}\cup\overline{a}$ which implies that $\{e\}\cup\overline{a}$ is closed under multiplication and forms a subgroup. $\Box$\\

For every $X=(\overline{h}, \mu)$ we have
\begin{align*}
S_{CX}=\frac{1}{\v G\v\cdot \v Z(h)\v}\sum_{k\in G} \tr_{\pi}(kh^{-1}k^{-1})\tr_{\mu}(e)=
\frac{\tr_{\pi}(h^{-1})\,\dim \mu\ }{\v Z(h)\v},
\end{align*}
and
\begin{align*}
S_{FX}= \frac{1}{\v Z(h)\v \cdot \v Z(a)\v} \sum_{khk^{-1}\in Z(a)}
\tr_{\mu}(k^{-1}a^{-1}k).
\end{align*}
Therefore, if $X\neq C, F$
\begin{align}\label{eq:tr-pi-h}
\frac{\tr_{\pi}(h^{-1})\,\dim \mu\ }{\v Z(h)\v} = \frac{1}{\v Z(h)\v \cdot \v Z(a)\v} \sum_{khk^{-1}\in Z(a)}
\tr_{\mu}(k^{-1}a^{-1}k).
\end{align}

\noindent(d) For any irreducible representation $\mu$ of $G$ different from $\pi$ we have
\begin{align*}
\tr_{\mu}(a^{-1})=\tr_{\mu}(a)=\dim \mu=\tr_{\mu}(e).
\end{align*}

Let $h=e$ in~\eqref{eq:tr-pi-h} and note that $k^{-1}a^{-1}k$ is a conjugate of $a^{-1}$ in $Z(h)=G$. Thus $\tr_{\mu}(k^{-1}a^{-1}k)=\tr_{\mu} (a^{-1})$ and
\begin{align*}
\frac{\dim \pi\,\dim \mu\ }{\v G\v} = \frac{\tr_{\mu}(a^{-1})}{\v Z(a)\v}.
\end{align*}
Then by~\eqref{eq:dim} we obtain $\tr_{\mu}(a^{-1})=\dim \mu$. $\Box$\\

\noindent(e) $\tr_{\pi}(h)=0$, for any $h\notin \{e\}\cup \overline{a}$, and $\tr_{\pi}(a)=\tr_{\pi}(a^{-1})=-1$.\\

The column $\overline{h}$ of the character table of $G$ is orthogonal to columns $e$ and $\overline{a}$. On the other hand, by (d) columns
$e$ and $\overline{a}$ coincide except at the representation $\pi$. Therefore, $\tr_{\pi}(h)=0$. $\tr_{\pi}(a)=-1$ can be shown using (b), and the orthogonality of $\pi$ and the trivial representation of $G$. $\Box$\\

\noindent(f) $\v Z(a)\v =\v \overline{a}\v +1$.\\

Because of the orthogonality of columns $e$ and $\overline{a}$ of the character table of $G$ we have
\begin{align*}
\sum_{\mu} \tr_{\mu}(e) \tr_{\mu}(a)^{\ast}=0,
\end{align*}
where the sum is over all irreducible representations of $G$. Thus $\sum_{\mu\neq \pi} (\dim \mu)^2 - \dim \pi =0$. On the other hand, we know that
$\sum_{\mu} (\dim \mu)^2 = \v G\v$. Therefore, $\v G\v - (\dim \pi)^2 - \dim \pi =0$ which by using $\dim \pi =\v a\v $ gives $\v Z(a)\v =\v \overline{a}\v +1$. $\Box$\\

\noindent(g)  $Z(a) = \{e\} \cup \overline{a}$.\\

According to (f) it is sufficient to show that $h\notin Z(a)$ for every $h\notin \{e\}\cup \overline{a}$.
This fact can easily be seen from~\eqref{eq:tr-pi-h} by letting $\mu=\mathbf{1}$. $\Box$\\

\noindent(h) $Z(a)\simeq \mathbf{F}_q^{+}$ where $q$ is a power of a prime number, and $\v G\v = \v Z(a)\v\cdot \v \overline{a}\v = q(q-1)$.\\

Since $Z(a)=\{e\} \cup \overline{a}$ is a normal subgroup, $Z(b)=Z(a)$ for every $b\in \overline{a}$. Thus $Z(a)$ is an abelian subgroup. On the other hand, the order of all elements of $\overline{a}=Z(a)\setminus e$ is the same. Therefore, $Z(a)$ is isomorphic to $\mathbf{F}_q^+$ where $q$ is a power of a prime number. $\Box$\\

For simplicity let $Z(a)=\mathbf{H}$. Then $\mathbf{H}$ is an abelian subgroup of $G$. We show that a multiplication $\times$ can be defined on $\mathbf{H}$ which together with the operation of $\mathbf{H}$ induced from $G$ turns it into a near-field. We then prove that $G\simeq \mathbf{H}\rtimes \mathbf{H}^{\times}$.

Since $\v G/\mathbf{H}\v = \v \overline{a}\v$, and $\mathbf{H}=Z(a)$, the cosets of $G/\mathbf{H}$ are in one-to-one correspondence with elements of $\overline{a}$; for every $b\in \overline{a}$ there exists a unique $\tilde x_b= x_b\mathbf{H} \in G/\mathbf{H}$ such that $x_b ax_b^{-1}=b$. Now define a binary operation $\times$ on $\mathbf{H}$ in the following form. $e\times b = b\times e=e$ for every $b\in \mathbf{H}$, and for $b, c\in \overline{a}$
\begin{align*}
b\times c = x_b x_c a x_c^{-1} x_b^{-1}.
\end{align*}
$\times$ is well-defined because elements of $\mathbf{H}$ commute with every element of $\bar{a}$.\\

\noindent(i) $\mathbf{H}^{\times}= (\mathbf{H}\setminus e, \times)$ is a group whose identity element is $a$. \\

The inverse of $b$ is $b'$ where $b'=x_b^{-1}a x_b$. The associativity is proved using $\tilde x_{b\times c} = \tilde x_b \tilde x_c$. $\Box$\\

\noindent(j) $\mathbf{H}$ with the induced operation from $G$ as the addition and $\times$ as the multiplication forms a near-field.\\

We need to show that multiplication is distributive from left with respect to addition: $b\times (cd) = (b\times c)(b\times d)$. If one of $b,c,d$ is equal to $e$, it obviously holds; otherwise both sides are equal to $x_b cdx_b^{-1}$. $\Box$\\

In the following we assume that $q=\v \mathbf{H}\v >2$ since otherwise $G\simeq \mathbf{H}\rtimes \mathbf{H}^{\times}$ is obvious.\\

\noindent(k) There exists $g\in G\setminus \mathbf{H}$ such that $G=\mathbf{H}Z(g)$.\\

Since $\mathbf{H}$ is a near-field, the centralizer of $\mathbf{H}^{\times}$ is non-trivial (see Section~\ref{sec:4.1}). This means that there exists $g\in G$ such that $gag^{-1}\neq a$ and $(gag^{-1})\times b = b\times (gag^{-1})$ for every $b\in \mathbf{H}$. In other words, for every $x\in G$, $gx a x^{-1}g = xgag^{-1}x^{-1}$, or equivalently, $\overline{g}\subseteq g\mathbf{H}$. Therefore, $\v \overline{g}\v \leq \v \mathbf{H}\v =q$, and then $\v Z(g)\v \geq q-1$. On the other hand, $\mathbf{H}$ is a normal subgroup of $G$, so $\mathbf{H}Z(g)$ is a subgroup and since $Z(g)\cap \mathbf{H}=Z(g)\cap Z(a)=\{e\}$, the size of this subgroup is equal to $q\v Z(g)\v$. Thus $\v Z(g)\v \leq q-1$ and therefore, $Z(g)$ is a subgroup of order $q-1$ and $G= \mathbf{H}Z(g)$. $\Box$\\

\noindent(l) $G\simeq \mathbf{H}\rtimes \mathbf{H}^{\times}$.\\

Since $G=\mathbf{H}Z(g)$ and $\mathbf{H}\cap Z(g)=\{e\}$, every element of $G$ can uniquely be written in the form of $bk$ where $b\in \mathbf{H}$ and $k\in Z(g)$. It is easy to see that the map which sends $bk\in G$ to $(b, kak^{-1})\in \mathbf{H}\rtimes \mathbf{H}^{\times}$ is an isomorphism. We are done.

\end{proof}

\subsection{Example: two modular invariants in $\mathcal{Z}(A_6)$}\label{sec:a6}

Assume that the transposition $(X, Y)$ forms a modular invariant in $\mathcal{Z}(G)$. Theorem~\ref{thm:uniqueness} classifies all groups for which there exists such a modular invariant when $X$ is a chargeon and $Y$ is a fluxion. If we relax this assumption by keeping $X$ to be a chargeon but assuming $Y=(\overline{a}, \rho)$ is arbitrary, most steps in the proof of Theorem~\ref{thm:uniqueness} (with some variations) still hold. In particular, $\dim \rho =1$  is enough to show that $G\simeq \near$.  (In this case proving (g) needs more work.)

There are two remaining cases. First, both $X$ and $Y$ are chargeon, and second, non of them is chargeon. The first case cannot happen; if $X=(e, \pi)$ and $Y=(e, \pi')$, $S_{\mathbf{0}X}=S_{\mathbf{0}Y}$ implies $\dim \pi=\dim \pi'$. Moreover, for every $g\neq e $, $S_{X(\overline{g}, \mathbf{1})}=S_{Y(\overline{g},\mathbf{1})}$ is equivalent to $\tr_{\pi}(g)=\tr_{\pi'}(g)$. Thus $\pi=\pi'$.

Now assume that $X=(\overline{a}, \rho)$ and $Y=(\overline{b}, \rho')$, and $a,b\neq e$. Then for every irreducible representation $\pi$ of $G$, $S_{X(e, \pi)} = S_{Y(e, \pi)}$, and we obtain
\begin{align*}
\frac{\tr_{\pi}(a^{-1})\dim \rho}{\v Z(a)\v} = \frac{\tr_{\pi}(b^{-1})\dim \rho'}{\v Z(b)\v}.
\end{align*}
For $\pi=\mathbf{1}$ we find that $\dim \rho/\v Z(a)\v = \dim \rho' /\v Z(b)\v$, and thus for every $\pi$, $\tr_{\pi}(a^{-1})=\tr_{\pi}(b^{-1})$. Equivalently, $a$ and $b$ belong to the same conjugacy class, and $X,Y$ have the same magnetic flux.

Here we present an example of a modular invariant in the latter case ($X=(\overline{a}, \rho)$ and $Y=(\overline{a}, \rho')$). Let $A_6$ be the alternating group of order six (the group of even permutations over $\{1, \dots, 6\}$).
Let $a=(1, 2)(3, 4)$. $\overline{a}$ is equal to the set of all permutations of the form $(t_1, t_2)(t_3, t_4)$. Then $\v \overline{a}\v = 45$, and $\v Z(a)\v= \v A_6\v/ \v \overline{a} \v = 8$;
\begin{align*}
Z(a)=\{ e, a, b_1=(1,2)(5,6), b_2=(3,4)(5,6), b_3=(1,3)(2,4), b_4=(1,4)(2,3),\nonumber\\ c_1=(1,3,2, 4)(5,6), c_2=(1,4,2,3)(5,6) \}.
\end{align*}
The conjugacy classes of $Z(a)$ are $\{e\}$, $\{a\}$, $\{b_1, b_2\}$, $\{b_3, b_4\}$, and $\{c_1, c_2\}$, and the character table of $Z(a)$ is as follows.
\begin{align*}
\begin{array}{|c|ccccc|}
\hline
       &  e   &   a    &    b_1,b_2   &    b_3,b_4   &    c_1,c_2    \\
\hline
\rho_1 &  1   &   1    &      1       &       1      &       1        \\
\rho_2 &  1   &   1    &      -1      &       -1     &       1         \\
\hline
\rho_3 &  1   &   1    &      1      &       -1     &       -1         \\
\rho_4 &  1   &   1    &      -1      &       1     &       -1         \\
\hline
\mu    &  2   &   -2   &      0      &       0     &       0         \\
\hline
\end{array}
\end{align*}
We claim that both transpositions $(X_1, X_2)$ and $(X_3, X_4)$, where $X_i=(\overline{a}, \rho_i)$, are modular invariants.

$T$ is invariant under these transpositions because $T_{X_i}=1$ for every $1\leq i\leq 4$.

For every $Y=(e, \pi)$, $S_{X_iY} = \tr_{\pi}(a^{-1})\dim \rho_i/\v Z(a)\v$, and since $\dim \rho_i =1$ for every $i$, $S_{X_iY}=S_{X_jY}$ for $i,j\in \{1, 2, 3, 4\}$.

For every $Y=(\overline{h}, \pi)$, where $h\notin \{e\}\cup \overline{a}$, we have
\begin{align*}
S_{X_iY} = \frac{1}{\v Z(a)\v\cdot \v Z(h)\v} \sum_{k:\, khk^{-1}\in Z(a)} \tr_{\rho_i}(kh^{-1}k^{-1}) \tr_{\pi}(k^{-1}a^{-1}k).
\end{align*}
Observe that $c_1, c_2$ are the only elements of $Z(a)$ which can be conjugates of $h^{-1}$, and $\tr_{\rho_i} (c_1) =\tr_{\rho_j}(c_1)$ for $i,j\in \{1,2\}$ and $i,j\in \{3,4\}$. Therefore, $S_{X_iY}=S_{X_jY}$.

Now it remains to show that $S_{X_1X_1}=S_{X_2X_2}$, $S_{X_3X_3}=S_{X_4X_4}$, $S_{X_1X_3}=S_{X_2X_3}=S_{X_2X_4}=S_{X_1X_4}$, and $S_{X_1Y}=S_{X_2Y}=S_{X_3Y}=S_{X_4Y}$, where $Y=(\overline{a}, \mu)$. These equalities can simply be verified given the character table of $Z(a)$.

\end{document}